\documentclass[8pt,letterpaper,twoside,english,british]{article}
\usepackage[centertags]{amsmath}
\usepackage{graphicx,indentfirst,amsmath,amsfonts,amssymb,amsthm,newlfont,longtable}

\usepackage[all]{xy}
\usepackage[usenames,dvipsnames]{color}
\usepackage{subcaption}

\DeclareMathOperator{\sech}{sech}
\DeclareMathOperator{\ess}{ess}
\DeclareMathOperator{\sign}{sign}

\usepackage[latin1]{inputenc}
\def\al{\alpha}

\def\epsi{\epsilon}

\def\de{\delta}

\def\l1{{\lambda}_1}

\newcommand{\f}{\frac}

\def\sign{\text{sgn}}
\def\ess{\text{ess sup}}
\def\supp{\text{supp}}
\def\Ai{\text{Ai}}
\def\Bi{\text{Bi}}

\newcommand{\w}{{\cal L}_w}

\newcommand{\pe}{\left\langle}
\newcommand{\pd}{\right\rangle}

\newcommand{\ds}{\displaystyle }
\newtheorem{definition}{Definition}[section]
\newtheorem{theorem}{Theorem}[section]
\newtheorem{exe}{Example}[section]
\newtheorem{lemma}{\bf Lemma}[section]
\newtheorem{corollary}{Corollary}[section]
\newtheorem{rem}{Remark}[section]

\newcommand{\p}{\partial}
\newcommand{\bb}{\begin{equation}}
\newcommand{\ee}{\end{equation}}
\newcommand{\ba}{\begin{array}}
\newcommand{\ea}{\end{array}}
\newcommand{\R}{\mathbb{R}}
\newcommand{\N}{\mathbb{N}}

\begin{document}
\pagenumbering{arabic}
\title{\huge \bf A family of wave equations with some remarkable properties}
\author{\rm \large Priscila Leal da Silva$^{1}$, Igor Leite Freire$^{2}$ and J\'ulio Cesar Santos Sampaio$^{2}$ \\
\\
\it $^1$ Departamento de Matem\'atica, UFSCar, Brazil\\
\rm E-mail pri.leal.silva@gmail.com\\
\\
\it $^2$ Centro de Matem\'atica, Computa\c c\~ao e Cogni\c c\~ao, UFABC, Brazil\\
\rm E-mails: igor.freire@ufabc.edu.br/igor.leite.freire@gmail.com and juliocesar.santossampaio@gmail.com}
\date{\ }
\maketitle
\vspace{1cm}
\begin{abstract}
We consider a family of homogeneous nonlinear dispersive equations with two arbitrary parameters. Conservation laws are established from the point symmetries and imply that the whole family admits square integrable solutions. Recursion operators are found to two members of the family investigated. For one of them, a Lax pair is also obtained, proving its complete integrability. From the Lax pair we construct a Miura-type transformation relating the original equation to the KdV equation. This transformation, on the other hand, enables us to obtain solutions of the equation from the kernel of a Schr\"odinger operator with potential parametrized by the solutions of the KdV equation. In particular, this allows us to exhibit a kink solution to the completely integrable equation from the 1-soliton solution of the KdV equation. Finally, peakon-type solutions are also found for a certain choice of the parameters, although for this particular case the equation is reduced to a homogeneous second order nonlinear evolution equation. 
\end{abstract}
\vskip 1cm
\begin{center}
{2010 AMS Mathematics Classification numbers:\vspace{0.2cm}\\
 35D30, 37K05\vspace{0.2cm} \\
Keywords: Evolution equations, recursion operators, Lax pair, integrability, Miura transformation, solitary wave solutions }
\end{center}
\pagenumbering{arabic}
\newpage
\section{Introduction}

A couple of years ago some authors considered the equation
\bb\label{1.1}
u_t+2a\frac{u_xu_{xx}}{u}=\epsilon au_{xxx},
\ee
with real constants $\epsilon$ and $a$, see \cite{sen}. In (\ref{1.1}), $u=u(x,t)$, where $(x,t)$ belongs to a convenient domain in $\R^2$. Such equation, with $\epsilon=a=1$, has in common with the KdV equation\footnote{In this case, by KdV equation we mean $u_t+u_{xxx}+6uu_{x}=0$, which is one of the forms of that equation. Through this paper other forms will be used depending on the convenience.} the fact that both admit the solution
$$
u(x,t)=\f{c}{2}\sech^{2}{\left(\f{\sqrt{c}}{2}(x-ct-x_0)\right)}
$$
where $c>0$ is the wave speed and $x_0$ is a constant. 

If one takes $u=e^w$, for a certain function $w=w(x,t)$, and substitute it into (\ref{1.1}), one arrives at (see \cite{sen})
\bb\label{1.2}
w_t=\epsilon a w_{xxx}+(3\epsilon-2)aw_xw_{xx}+(\epsilon -2)aw_{x}^3.
\ee

For $\epsilon=2/3$ equation (\ref{1.2}) is reduced to the potential mKdV equation, while, after differentiating (\ref{1.2}) with respect to $x$ and next defining $v:=w_x$, one arrives at the following equation
\bb\label{1.3}
v_{t}=\epsi a v_{xxx}+\f{3\epsi-2}{2}a(v^2)_{xx}+(\epsi-2)a(v^3)_x.
\ee
Therefore, if $v$ is a solution of equation (\ref{1.3}), then $u=e^{D_{x}^{-1}v}$ is a solution of (\ref{1.1}), where $D_{x}^{-1}$ is the formal inverse of total derivative operator. For some properties and deep discussion on the operator $D_x$ and its inverse, see \cite{SW2001}.

Another interesting relation of (\ref{1.1}) can be obtained as follows: if $\epsilon=-2/3$, then (\ref{1.1}) can be rewritten as $(u^2)_t+2a/3(u^2)_{xxx}=0$, see \cite{sen}. Therefore, defining $\rho:=u^2$, the latter equation is equivalent to the Airy's equation
\bb\label{1.4}
\rho_{t}+\f{2}{3}a\rho_{xxx}=0,
\ee
which is nothing but the linear KdV equation. A summary of some related KdV-type equations and (\ref{1.1}) is presented in Figure \ref{fig1}.

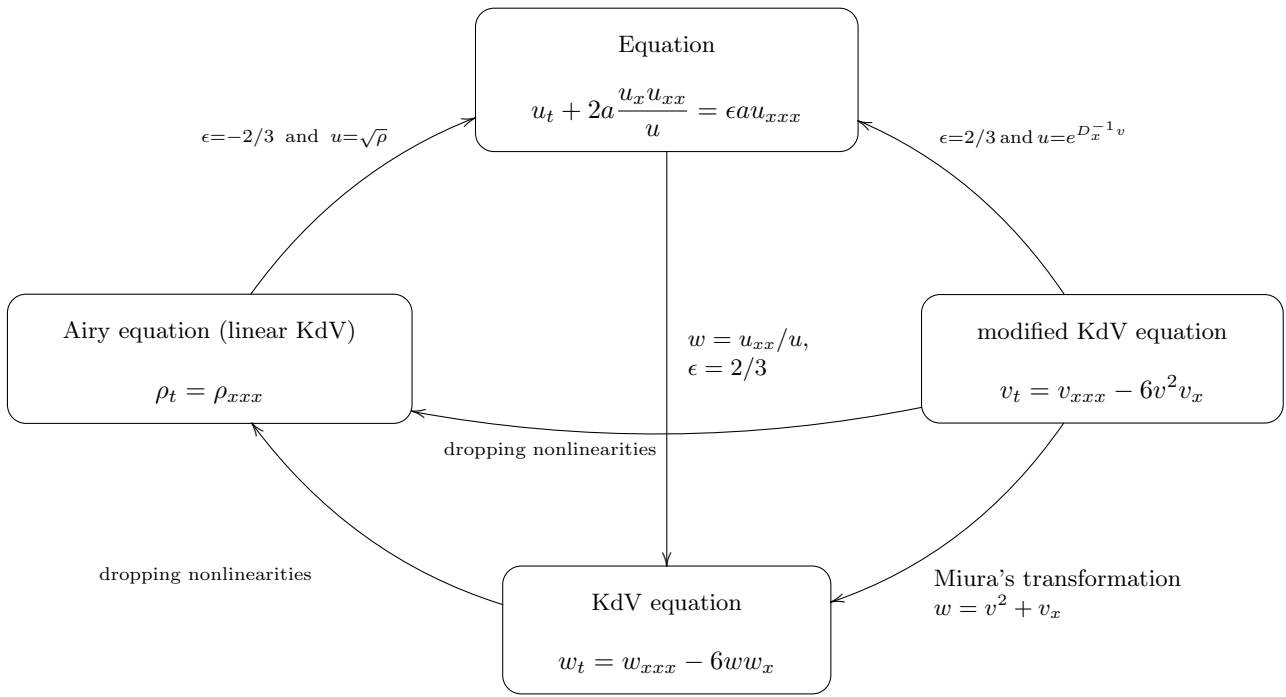
\begin{figure}[h!]
\centering
\xymatrix{ 
&*++[F-:<7pt>]{\ba{lcl}&\small\txt{Equation}&\\ \\ &\ds{u_t+2a\frac{u_xu_{xx}}{u}=\epsilon au_{xxx}}\ea}\ar[dddd]^{\small{\ba{l} w=u_{xx}/u,\\\epsi=2/3\ea}}\quad\quad\quad&\\
\\
*++[F-:<7pt>]{\ba{lcl}&\small\txt{Airy equation (linear KdV)}&\\ \\ &\ds{\rho_t=\rho_{xxx}}\ea} \quad\ar@/^1.0cm/[uur]^{\epsi=-2/3\,\,\,\text{and}\,\,\,u=\sqrt{\rho}\quad \quad}  &&*++[F-:<7pt>]{\ba{lcl}&\small\txt{modified KdV equation}&\\ \\ &\ds{v_t=v_{xxx}-6v^2v_x}\ea}\ar@/_1.0cm/[uul]_{\quad\quad\epsi=2/3\,\text{and}\, u=e^{D_x^{-1}v}}\ar@/^1.0cm/[ddl]^{\quad\quad\small{\ba{l}\text{Miura's transformation}\\ w=v^2+v_x\ea}} \ar@/^1.0cm/[ll]^{\text{dropping nonlinearities}\quad\quad\quad\quad\quad\quad\quad\quad}\\
\\
&*++[F-:<7pt>]{\ba{lcl}&\small\txt{KdV equation}&\\ \\ &\ds{w_t=w_{xxx}}-6ww_x\ea}\ar@/^1.0cm/[uul]^{\text{dropping nonlinearities}\quad\quad\quad\quad\quad\quad}&}

\caption{\small{The figure shows some connections among KdV type equations. The relation between KdV and mKdV and the Airy equation is formal, in the sense that the KdV and mKdV are reduced to Airy equation if the nonlinear effects can be neglected. The other relations apply solutions of the original equation (the start point of the arrow) into the target equations (at the end of the arrow). In the transformations regarding equation (\ref{1.1}), it is chosen $a=1/\epsi$. Transformation $\rho=u^2$ transforms the Airy equation into equation (\ref{1.1}). On the other hand, if $v$ is a solution of the mKdV equation, then $u=e^{D_x^{-1}}$ is a solution of (\ref{1.1}) when $\epsi=-2/3$. It is quite similar to the Miura's transformation, that takes a solution $v$ of mKdV and transforms it into a solution of KdV. Therefore, solutions of the mKdV equation can be transformed into solutions of the KdV and (\ref{1.1}). The transformation relating equation (\ref{1.1}) and the KdV shall be discussed in section \ref{secmiu}.}}\label{fig1}
\end{figure}

In spite of being nonlinear, equation (\ref{1.1}) has some intriguing and interesting properties regarding its solutions due to its homogeneity, that is, invariance under transformations $(x,t,u)\mapsto (x,t,\lambda u)$, $\lambda=const$. A very simple observation is the fact that if one assumes $u(x,t)=e^{i(kx-\omega t)}$, where $i=\sqrt{-1}$ and $\omega=\omega(k)$, then we conclude that (\ref{1.1}) has plane waves\footnote{In \cite{sen} the authors considered plane waves when $a=1$.} provided that $k$ is a real number and
\bb\label{1.5}
\omega(k)=(\epsi-2)ak^3.
\ee
In this case, the phase velocity is $c=\omega/k$. Sometimes it will be more convenient to write the wave number as a function of $c$, that is,
\bb\label{1.6}
k^2=\f{c}{a(\epsi-2)}.
\ee

Additionally, if $k^2<0$, then we have exponential solutions given by (see \cite{juliotese,ijaims})
\bb\label{1.7}
u_{\pm}(x,t)=e^{\pm\sqrt{\f{c}{a(2-\epsi)}}(x-ct)}.
\ee

Solutions (\ref{1.7}) are not defined for $\epsi=2$ or $a=0$. The first problem can be overcome invoking and understanding what (\ref{1.5}) implies, from which one concludes that the plane waves are reduced to stationary solutions if $\epsi=2$. The problem related to $a=0$ can be avoided due to a very simple argument: if one takes $a=0$ into (\ref{1.1}), then one would obtain a totally uninteresting equation. Therefore, through this paper we make the hypothesis $a\neq0$. The reader might argue that one could rescale (\ref{1.1}) and eliminate $a$, but it shall be convenient in the remaining sections to leave both constants $a$ and $\epsi$ in (\ref{1.1}). The reader shall have the opportunity to observe in our analysis that the cases $\epsi=1/a$ and $\epsi=0$ will be of great interest and, for these reasons, we prefer to leave these constants as they are in (\ref{1.1}).

 An interesting observation can now be made: fixing $\sign(a)=\sign(c)$, when $c\neq0$, then one can note that the behaviour of travelling wave solutions $u=\phi(x-ct)$ changes depending on whether $\epsi<2$ or $\epsi>2$, see equations (\ref{1.6}) and (\ref{1.7}), for example. In section \ref{sol} we shall investigate solutions of (\ref{1.1}), mainly the bounded and weak ones.

These values of $\epsi$ are not surprising. Actually, this paper is motivated by some observations made by the authors of \cite{sen} in that paper. At the very beginning, on page 4118, they made the following comment:
\begin{center}
\begin{flushleft}
\small{\emph{We did not find any value of $\epsi$ at which a transformation reduces the SIdV to the KdV equation. But there are special values of $\epsi$ at which it comes close.}}
\end{flushleft}
\end{center}
By SIdV the authors of \cite{sen} named equation (\ref{1.1}). We, however, do not use this name in our paper\footnote{We actually do not agree with the acronym. Our results show that the equation is really strongly related to the KdV equation, as the authors of \cite{sen} suspected, but our discoveries in the present paper also support the view that the equation itself has other very interesting properties that make it of interest by its own properties.}. After some lines, they pointed out the following:
\begin{center}
\begin{flushleft}
\small{\emph{It also indicates that $2,\,2/3$ are somewhat special. At these values, we get KdV-like dispersive wave equations with advecting velocities $\propto w_x^2$ and $w_{xx}$. These are among the simplest PT symmetric advecting velocities beyond KdV. Moreover, at $\epsi= 2/3$, the sign of the `local diffusivity' is reversed.}}
\end{flushleft}
\end{center}

These observations made our interest arise. Moreover, they were quite stimulating and nearly a challenge. For these reasons, we investigate equation (\ref{1.1}) from several point of views, looking for new solutions, conservation laws, integrability properties and further relations with the KdV equation. 

For instance, in \cite{sen} the authors found two conservation laws for (\ref{1.1}) for arbitrary values of $\epsi$. Later, in \cite{ijaims}, two of us have established a new conservation law for $\epsi=-2$. In the present paper we find new conservation laws for $\epsi=0$ and $\epsi=-2/3$ obtained from the point symmetries of (\ref{1.1}) and the results proved in \cite{ib2,ib6,ijcnsns}. This is done in section \ref{cl}. Furthermore, our results support the viewpoint of the authors of \cite{sen}, in which they claimed that cases $\epsi=2,\,2/3$ would be special. Additionally, our investigation on conservation laws shall show that the list of special cases can be enlarged to the cases $\epsi=0,-2,-2/3$.

From the results established in \cite{PEE2004,SW1998}, we are able to find recursion operators for the cases $\epsi=\pm2/3$. These recursion operators combined with the point symmetry generators provide and infinite hierarchy of higher order symmetries, which suggest that these cases might be integrable. Then, in section \ref{integrability}, we not only present the recursion operators, but we also construct a Lax pair for the case $\epsi=2/3$ and find higher order conserved densities for the case $\epsi=-2/3$. These cases reinforce the observations pointed out in the last paragraph, as well as that one made in \cite{sen} and, as a consequence of the Lax pair, we shall be able to prove that in the case $\epsi=2/3$ equation (\ref{1.1}) can be transformed to the KdV by a Miura type transformation, as already presented in the Figure \ref{fig1}.

The Lax pair is a serendipitous discovery, because it firstly lets us respond to one of the observations made in \cite{sen}. Secondly, it enables us to construct solutions of the KdV equation from solutions of (\ref{1.1}) with $\epsi=2/3$. Thirdly, it lets us construct solutions of the latter equation from solutions of the first by looking to functions belonging to the kernel of a Schrödinger operator with potential parametrized by the solutions of the KdV equation. More interestingly, in this situation the solutions of our equation admit a sort of superposition of its solutions. This is explored in section \ref{secmiu}.

The fact that the solutions of the KdV equation can lead to solutions of (\ref{1.1}) when $\epsi=2/3$ and $\epsi a=1$ is used to construct a kink wave solution of the latter equation, as shown in section \ref{sol}.

In section \ref{other equations} we make further comments on other results obtained in \cite{sen} with respect to other equations sharing the same $\sech^2$ solution of the KdV.

A discussion on the results of this paper is presented in section \ref{other equations}, while in section \ref{conclusion} we summarise the results of the paper in our concluding section.

In the next section we introduce some basic facts regarding notation and previous results we shall need in the next sections. 

\section{Notation and conventions}\label{notation}

Given a differential equation $F=0$, by ${\cal F}$ we mean the equation $F=0$ and all of its differential consequences. Here we assume that the independent variables are $(x,t)$ and the dependent one is $u$. We shall avoid details on point symmetries once we assume that the reader is familiar with Lie group analysis. We, however, guide the less familiarised reader to the references \cite{2ndbook,i,ol} and \cite{OW2000}.

For our purposes it is more convenient to assume that any (point, generalised or higher order) symmetry generator is on evolutionary form
\bb\label{2.1}
\textbf{v}=Q[x,t,u_{(n)}]\f{\p}{\p u},
\ee
where $Q[x,t,u_{(n)}]$ means that $Q$ is a function of the independent variables $(x,t)$, dependent variable $u$ and derivatives of $u$ up to a certain order $n$. In what follows we only write $Q$ instead of $Q[x,t,u_{(n)}]$. For further details, see Olver \cite{ol}, chapter 5. 

Given a point symmetry generator 
$$\textbf{v}=\tau(x,t,u)\f{\p}{\p t}+\xi(x,t,u)\f{\p}{\p x}+\eta(x,t,u)\f{\p}{\p u}$$
of a given differential equation $F=0$, it is equivalent to an evolutionary field (\ref{2.1}) if one takes $Q=\eta-\xi u_x-\tau u_t$. For instance, in \cite{juliotese,ijaims} it was proved that the (finite group of) point symmetries of (\ref{1.1}) are generated by the continuous transformations $
(x,t,u)\mapsto(x+s,t,u)$, $(x,t,u)\mapsto(x, t+s,u)$, $(x,t,u)\mapsto(x,t,e^s u)$, and $(x,t,u)\mapsto(e^s x,e^{3s}t,u)$, where $s$ is a continuous real parameter. In particular, their corresponding evolutionary generators are
\bb\label{2.2}
\textbf{v}_{1}=u_x\f{\p}{\p u},\,\,\,\textbf{v}_{2}=u_t\f{\p}{\p u},\,\,\,\textbf{v}_{3}=(3tu_t+xu_x)\f{\p}{\p u}\,\,\text{and}\,\,\textbf{v}_{4}=u\f{\p}{\p u}.
\ee

By $K[u]$ we mean a function depending only on $u$ and its derivatives with respect to $x$. The Fréchet derivative of $K[u]$, denoted by $K_\ast$, is the operator defined by
$$
K_{\ast}v:=\left.\f{d}{ds}\right|_{s=0}K[u+s v].
$$

Given two functions $K[u]$ and $Q[u]$, its commutator is defined by $[K,Q]:=Q_{\ast}[K]-K_\ast[Q].$
In particular, an operator (\ref{2.1}) is a symmetry of an evolution equation 
\bb\label{2.2}
u_t=K[u]
\ee
if $Q_{t}=[K,Q]$. Moreover, for equations of the type (\ref{2.2}) one can always express the components of any evolutionary symmetry in terms of $t,x,u$ and derivatives of $u$ with respect to $x$ using relation (\ref{2.2}) and its differential consequences, although it might not be convenient in all situations, as one can confirm in section \ref{integrability}.

It shall be of great convenience to recall some notation on space functions. We avoid a general presentation, although we invite the readers to consult \cite{brezis}, Chapter 8; \cite{hunter}, Chapter 11; and \cite{schwartz}, Chapter 2, for a deeper discussion.

To begin with, consider $\N:=\{0,\,1,\,2,\,3,\cdots\}$, $\al=(\al_1,\cdots,\al_n)\in\N^n$, $|\al|:=\al_1+\cdots+\al_n$ and $\Omega\subseteq\R^n$ an open set.  The support of a function $f:\Omega\rightarrow\R$, denoted by $\supp{(f)}$, is the closure of the set $\{x;\,f(x)\neq0\}$. The set of all functions $f:\Omega\rightarrow\R$ with compact support, having continuous derivatives of every order, is either usually denoted by $C^\infty_{0}(\Omega)$ or ${\cal D}(\Omega)$. A member of $C^\infty_{0}(\Omega)$ is called test function and the space of distributions defined on $\Omega$ is referred to ${\cal D}'(\Omega)$. Given $T\in {\cal D}'(\Omega)$ and $\phi$ a test function, the (distributional or weak) derivative $D^\al T$ is defined by $\pe D^\al T,\phi \pd:=(-1)^{|\al|}\pe T,D^\al\phi\pd$, where
$$
D^{\al}\varphi:=\f{\p^{|\al|}\varphi(x)}{\p (x^1)^{\al_{1}}\cdots \p (x^n)^{\al_{n}}}.
$$

By $L^p(\Omega)$, $1\leq p\leq\infty$, we mean the classical space of (equivalence classes of) integrable functions endowed with the norm 
$$
\|f\|_{L^p}:=\left\{
\ba{l}
\ds{\left(\int_{\Omega}|f(x)|^pdx\right)^{\f{1}{p}},\quad if \quad1\leq p<\infty},\\
\\
\ess |f(x)|,\quad if \quad p=\infty.
\ea\right.
$$
If $w:\Omega\rightarrow\R$ is a nonnegative integrable function, then the weighted $L^p(\R,w)$ space,$\,1\leq p<\infty$, is the vector space of equivalence classes of integrable functions endowed with the norm
$$\|f\|_{(L^p,w)}:=\left(\int_{\Omega}w(x)|f(x)|^pdx\right)^{\f{1}{p}}.$$

The Sobolev space $W^{k,p}(\Omega)$ is the space of functions $f:\Omega\rightarrow\R$ such that $f,\,D^\al f\in L^p(\Omega)$, for $|\al|=1, 2,\cdots, k$, where $D^\al f$ means the $\al-th$ distributional derivative of $f$. On $W^{k,p}(\Omega)$ we shall use the norm
$$
\|f\|_{W^{k,p}}:=\left\{
\ba{l}
\ds{\left(\int_{\Omega}\left(|f(x)|^p+|Df(x)|^p+\cdots+|D^k f(x)|^p\right)dx\right)^{\f{1}{p}},\quad if \quad1\leq p<\infty},\\
\\
\ess |f(x)|+\ess |Df(x)|+\ess |D^kf(x)|,\quad if \quad p=\infty.
\ea\right. 
$$

Let $K\subseteq\Omega$ be a compact set and $\chi_{K}$ be the characteristic function of $K$. The local $L^p(\Omega)$ space, denoted by $L^p_{loc}(\Omega)$, $1\leq p\leq\infty$, is the set of functions $f:\Omega\rightarrow\R$ such that $f\chi_K\in L^p(\Omega)$, for every compact set $K\subseteq\Omega$. Analogously we can also define the local Sobolev spaces $W^{k,p}_{loc}(\Omega)$.


\section{Conservation laws derived from point symmetries}\label{cl}

In this section we establish conservation laws for equation (\ref{1.1}) derived from point symmetries using Ibragimov's machinery \cite{ib2} and some results obtained in \cite{ijcnsns}. We begin with the following lemma (the summation over repeated indices is presupposed).
\begin{lemma}\label{lema1}

Let 
$$
X=\xi^{i}\f{\p}{\p x^i}+\eta\f{\p}{\p u}
$$
be any symmetry (Lie point, generalised) of a given differential equation 
\bb\label{3.1}F(x,u,u_{(1)},\cdots,u_{(s)})=0\ee
and
\bb\label{3.2}
F^{\ast}(x,u,v,\cdots,u_{(s)},v_{(s)}):=\f{\de}{\de u}{\cal L}=0,
\ee
where ${\cal L}=vF$ is the formal Lagrangian, $\delta/\delta u$ is the Euler-Lagrange operator and $(\ref{3.2})$ is the adjoint equation to equation $(\ref{3.1})$. Then the combined system $(\ref{3.1})$ and $(\ref{3.2})$ has the conservation law $D_{i}C^{i}=0$, where
\bb\label{3.3}
\ba{lcl}
C^{i}&=&\ds{\xi^{i}{\cal L}+W\,\left[\f{\p{\cal L}}{\p u_{i}}-D_{j}\left(\f{\p{\cal L}}{\p u_{ij}}\right)+D_{j}D_{k}\f{\p{\cal L}}{\p u_{ijk}}-\cdots\right]}\\
\\
&&\ds{+D_{j}(W)\,\left[\f{\p{\cal L}}{\p u_{ij}}-D_{k}\left(\f{\p{\cal L}}{\p u_{ijk}}\right)+\cdots\right]}\ds{+D_{j}D_{k}(W)\,\left[\f{\p{\cal L}}{\p u_{ijk}}-\cdots\right]+\cdots}
\ea
\ee
and $W=\eta-\xi^{i}u_{i}$.
\end{lemma}
\begin{proof}
See \cite{ib2}, Theorem 3.5.
\end{proof}

Given a differential equation (\ref{2.2}), a vector field $A=(A^0,A^1)$ is called a trivial conservation law for the equation if either $D_{t}A^{0}+D_{x}A^{1}$ is identically zero or if its components vanish on ${\cal F}$. Two conserved vectors $B$ and $C$ are equivalent if there exists a trivial conserved vector $A$ such that $C=A+B$. In the last situation one writes $A\sim C$ and it defines an equivalence relation on the vector space of the conserved vectors of a given differential equation. We invite the reader to consult \cite{pop3} and references therein for further details.

It follows from Lemma \ref{lema1} that if (\ref{2.1}) is a (generalized) symmetry generator of (\ref{1.1}), then the vector $C=(C^0,C^1)$, with components 
\bb\label{3.1.1}
C^{0}=vQ,\quad
C^{1}=Q\left(2a\frac{u_x^2}{u^2}v-2a\frac{u_x}{u}v_x - \epsi a v_{xx}\right) + D_xQ \left(2a\frac{u_x}{u}v + \epsi a v_x\right) - \epsi a v D_x^2Q
\ee
provides a nonlocal conservation law for (\ref{1.1}), see \cite{ib2,ib6,ijcnsns} for further details.

A natural observation from Lemma \ref{lema1} and the components (\ref{3.1.1}) is that the vector established relies upon the variable $v$ and, therefore, it does not provide a conservation law for equation (\ref{1.1}) itself, but to it and its corresponding adjoint. The following definition can be of great usefulness for dealing with this problem.
\begin{definition}\label{defIbra}
A differential equation $(\ref{3.1})$ is said to be nonlinearly self-adjoint if there exists a substitution $v=\phi(x,u)$, with $\phi(x,u)\neq0$, such that
\bb
\left.F^{\ast}\right|_{v=\phi}=\lambda(x,u,\cdots)F,
\ee
for some $\lambda\in{\cal A}$.
\end{definition}

One should now check if equation (\ref{1.1}) is nonlinearly self-adjoint. The following result, proved in \cite{ijcnsns} (Theorem 2), helps us with equation (\ref{1.1}).

\begin{lemma}\label{lema2}
Equation 
\bb\label{3.5}
u_{t}+f(t,u)u_{xxxxx}+r(t,u)u_{xxx}+g(t,u)u_{x}u_{xx}+h(t,u)u_{x}^{3}+a(t,u)u_{x}+b(t,u)=0,
\ee
is nonlinearly self-adjoint if and only if there exists a function $v=\phi(x,t,u)$ such that the coefficient functions of $(\ref{3.5})$ and the function $\phi$ satisfy the constraints 
\bb\label{3.6}
\ba{l}
(\phi f)_{uu}=0,\,\,(\phi f)_{xu}=0,\,\, (\phi h)_x=(\phi r)_{xuu},\\
\\
(\phi b)_u-\phi_t-\phi_xa-\phi_{xxx}r-\phi_{xxxxx}f=0, \\
\\
(\phi g)_x=3(\phi r)_{xu}, \,\,\,
2\phi h-(\phi g)_u+(\phi r)_{uu}=0.
\ea
\ee
\end{lemma}

We have at our disposal all needed ingredients to prove the following result.

\begin{theorem}\label{main1}
Equation $(\ref{1.1})$ is nonlinearly self-adjoint and its corresponding substitutions are given by:
\begin{itemize}
\item If $\epsilon \neq \pm2,\,\,-2/3,\,\,0$, then we have
$v=c_1u+c_2u^{-\frac{2}{\epsilon}}$, where $c_1$ and $c_2$ are arbitrary constants.

\item If $\epsilon =2$, then $v=c_1u+c_2\ln |u|$, where $c_1$ and $c_2$ are arbitrary constants.
\item If $\epsilon =-2$, then
$v=c_1u+c_2u\ln |u|$, where $c_1$ and $c_2$ are arbitrary constants.
\item If $\epsilon=-2/3$,
$v=\rho u+c_2u^3$, where $c_2$ is an arbitrary constant, and $\rho$ is a solution of the Airy equation $(\ref{1.4})$.

\item If $\epsi=0$, then $v=u$.
\end{itemize}
\end{theorem}

\begin{proof}
Since equation (\ref{1.1}) can be obtained from (\ref{3.5}), the result follows from Lemma \ref{lema2} by substituting $f(t,u)=h(t,u)=a(t,u)=b(t,u)=0$, $r(t,u)=-\epsi a$ and $g(t,u)=2a/u$, into (\ref{3.6}) and solving the system to $\phi$.
\end{proof}

The reader might be tempted to ask whether one could try to find differential substitutions for equation (\ref{1.1}) or not. Roughly speaking, the answer is yes. However, from Theorem 2 of \cite{rita} we can assure that substitutions depending on the derivatives of $u$, if they exist, must be of order equal or greater than 2, and looking for such higher order substitutions goes far beyond our purposes in this paper.

Assume that $C=(C^0,C^1)$ is a conserved vector for an evolution equation (\ref{2.2}) and $u,\,u_x\,u_{xx},\,u_{tx},\cdots\rightarrow0$ when $|x|\rightarrow\infty$. Then $C^0$ is a conserved density while the corresponding $C^1$ is the conserved flux for (\ref{2.2}) and the quantity 
$$H[u]=\int_{\R}C^0dx,$$
called {\it first integral}, is conserved along time.

From Theorem \ref{main1} and equation (\ref{3.1.1}), the quantity
\bb\label{3.1.2}
H[u]=\int_{\R}uQ\, dx
\ee
is conserved for any solution of equation (\ref{1.1}), notwithstanding the values of $\epsi$. In the next table we present low order conservation laws for equation (\ref{1.1}) using the classification above.
\begin{center}
\begin{longtable}{|l||l||l||l||l|}
\caption{\small{Conserved vectors of equation (\ref{1.1}). Above, the function $\rho$ is any solution of equation (\ref{1.4}). The vector fields are the evolutionary forms (\ref{2.2}) of the point symmetries of equation (\ref{1.1}). The first two conservation laws were first obtained in \cite{sen} and later in \cite{ijaims}. In this last reference it was found the third conservation law, using the direct method \cite{baprl,abeu1,abeu2}. The remaining conservation laws are new.}}\\
\hline\label{tab1}
 $\epsi$            &     Substitution $v$         & Generator     & Conserved density $C^0$ & Conserved flux $C^1$\\\hline
$\forall$   &   $u$                        & $\textbf{v}_3,\,\,\textbf{v}_4$           &  $u^2$& $(2+\epsi)au_{x}^2-2\epsi a uu_{xx}$  \\\hline
$\neq0$   &   $u^{-\f{2}{\epsi}}$   & $\textbf{v}_3,\,\,\textbf{v}_4$           & $u^{\f{\epsi-2}{\epsi}}$ & $(2-\epsi)au^{-\f{2}{\epsi}}u_{xx}$   \\\hline
$\epsi=-2$         &  $u\ln{|u|}$            & $\textbf{v}_3$ & $2u^2\ln{|u|}-u^2$ &  $8au\ln{|u|}\,u_{xx}-4au_{x}^2$  \\\hline
$\epsi=-2$         &  $u\ln{|u|}$            & $\textbf{v}_4$ & $u^2\ln{|u|}$ &  $-2au_{x}^2+2auu_{xx}+4au\ln{|u|}\,u_{xx}$  \\\hline

$\epsi=2$         &  $\ln{|u|}$                & $\textbf{v}_3$ & $\ln{|u|}$ &  $-\f{2a}{u}u_{xx}$  \\\hline

$\epsi=-2/3$         &  $\rho u$                & $\textbf{v}_1$ & $\rho_{x}\,u^2$ &  $-\rho_tu^2+\frac{4}{3}\rho_{x}u^{2}_{x}-\frac{4}{3}a\rho_{xx}uu_{x}+\frac{4}{3}a\rho_{x}uu_{xx}$  \\\hline

$\epsi=-2/3$         &  $\rho u$                & $\textbf{v}_2$ & $\rho_{xxx}u^2$ &  $-2\rho_tu^2_x2-2\rho_tuu_{xx}+2\rho_{xt}uu_x-\rho_{xxt}\,u^2$  \\\hline

$\epsi=-2/3$         &  $\rho u$                & $\textbf{v}_4$ & $\rho u^2$  & $-\frac{4}{3}a\rho_xuu_x+\frac{4}{3}a\rho u^2_x+\frac{2}{3}a\rho_{xx}u^2+\frac{4}{3}a\rho uu_{xx}$  \\\hline
\end{longtable}
\end{center}

We would like to observe some facts about the conservation laws of (\ref{1.1}).
\begin{enumerate}
\item In all cases the $L^2(\R)-$norm of the solutions are conserved. Actually, we have the first integral
\bb\label{3.1.3}
H_0[u]=\int_{\R} u^2dx,
\ee
which is equivalent to $H_0[u]=\|u\|^2_{L^2(\R)}$.

\item For the case $\epsi=-2/3$ we can conclude, from the Hölder inequality, that $u\in L^3(\R)$. Actually, on the one hand, we have $u\in L^2(\R)$. On the other hand, from the substitution $v=u^3$ and the generator $\textbf{v}_4$ (second row of the table), we have the following first integral:
\bb\label{3.1.4}
H_1[u]=\int_{\R} u^4dx.
\ee

On the other hand, $\|u^3\|_{L^1(\R)}\leq \|u\|_{L^2(\R)}\|u^2\|_{L^2(\R)}$, and then
$$
\|u\|^3_{L^3}(\R)=\int_{\R}|u(x,t)|^3dx=\|u^3\|_{L^1(\R)}<\infty.
$$

\item With respect to case $\epsi=-2$ we have another interesting fact concerning first integrals. From the substitution $v=u\ln{|u|}$ and the generator $\textbf{v}_4$ we obtain the conserved quantity
\bb\label{3.1.4}
H_2[u]=\int_{\R}u^2\ln{|u|} \,dx
\ee
for all (rapidly decreasing) positive solutions of (\ref{1.1}).

Under the change $u=\sqrt{\rho}$, with $\rho>0$, we can alternatively write (\ref{3.1.4}) as
$$\bar{H}_2=\int_{\R} \rho\ln{\rho}\,dx.$$

In the later case, $\rho$ is a (rapidly decaying) positive solution of the Airy's equation (\ref{1.5}).

The integrand $u\ln{|u|}$ is related to logarithmic Sobolev inequalities in measure spaces as the follows.
Let $d\gamma(x)=(2\pi)^{-1/2}e^{-x^2}dx$ be the Gaussian measure and $f:\R\rightarrow\R_{+}$ be a positive, smooth function. Then $d\mu=fd\gamma$ is a measure of probability and
$$
H(\mu|\gamma):=\int_{\R}f\,\ln{f}d\gamma
$$
and
$$I(\mu|\gamma):=\int_{\R}\f{|f'|^2}{f}d\gamma$$
are the relative entropy of $d\mu=fd\gamma$ with respect to $\gamma$ and the Fisher information of $\mu$ with respect to $\gamma$. Then the logarithm Sobolev inequality 
$$
\int_{\R}f\,\ln{f}d\gamma\leq\int_{\R}\f{|f'|^2}{f}d\gamma
$$
indicates that for every probability, the relative entropy is upper limited by the Fisher information, that is, $H(\mu|\gamma)\leq I(\mu|\gamma)$. For further details, see \cite{adams,gross,faruk1,ledoux}.

\item Yet about case $\epsi=-2/3$, we have first integrals like
\bb\label{3.1.5}
H_3[u]=\int_\R \rho\,u^2\,dx.
\ee

If the function $\rho=\rho(x,t)$ is non-negative and non-identically vanishing satisfying (\ref{1.4}), then the solutions $u$ satisfying (\ref{3.1.5}) are members of a sort of weighted $L^2(\R,\rho)$ space. For instance, the function $\rho(x,t)=x^2$ is a non-negative function everywhere and if $u$ is a solution satisfying (\ref{3.1.5}), then $u(\cdot,t)\in L^{2}(\R,x^2)$. 
\end{enumerate}

\section{Integrable members}\label{integrability}

In this section we find recursion operators and a Lax pair for a certain case of equation (\ref{1.1}). In some sense, to be better understood by the end of this section, these conditions shall enable us to find integrable members of (\ref{1.1}).

\subsection{Recursion operators and Lax pair}

A pseudo-differential operator $\mathfrak{R}$ is said to be a recursion operator of (\ref{1.1}) if and only if every evolutionary symmetry (\ref{2.1}) of (\ref{1.1}) is taken into another evolutionary symmetry 
$\textbf{u}=\tilde{Q}\frac{\p}{\p u},$ where $\tilde{Q}=\mathfrak{R}\,Q$. A complete treatment on recursion operators can be found in \cite{oljmp} and  in chapter 5 of \cite{ol}.

In \cite{SW1998}, section 8, the authors determined whether the equation
\bb\label{3.2.1}
u_t = u_{xxx}+3f(u)u_xu_{xx}+g(u)u_x^3
\ee
admits a recursion operator. According to that reference, if the functions $f$ and $g$ satisfy the conditions
\bb\label{3.2.2}
\frac{\p g}{\p u} = \frac{\p^2 f}{\p u^2}+2fg-2f^3,\quad g=\frac{\p f}{\p u}+f^2+h, \quad \frac{\p h}{\p u}=2fh,
\ee
for a certain function $h=h(u)$, then it admits a recursion operator
\bb\label{3.2.3}
\mathfrak{R}_h = (D_x+f(u)u_x +2u_xD_x^{-1}h(u)u_x)(D_x+f(u)u_x).
\ee

Equation (\ref{1.1}) can be put into the form (\ref{3.2.1}) if one takes $\epsilon a=1$, $g=0$ and $f(u)=-2a/(3u)$. In this case, the constraints (\ref{3.2.2}) read $\epsilon = \pm 2/3$ and $\ds{h(u) = -\frac{2a}{3u^2}\left(1+\frac{2}{3}a\right)}$. The last condition furnishes $h=0$ when $\epsilon =-2/3$ and $h=-2u^{-2}$ when $\epsilon = 2/3$.

Substituting these values of $h$ into (\ref{3.2.3}) one proves the following:
\begin{theorem}\label{teo4.1}
Equation $(\ref{1.1})$, with $\epsi a=1$, admits the recursion operators
\bb\label{3.2.4}
\mathfrak{R}^- = D_x^2 + 2u^{-1}u_xD_x+u^{-1}u_{xx},
\ee
for $\epsi=-2/3$, and
\bb\label{3.2.5}
\mathfrak{R}^+ = D_x^2 - 2u^{-1}u_xD_x-u^{-1}u_{xx} +u^{-2}u_x +u_xD_x^{-1}(u^{-2}u_{xx}-u^{-3}u_x^2),
\ee
for $\epsi=2/3$.
\end{theorem}

Theorem \ref{teo4.1} can be also proved from the results established in \cite{PEE2004}, see Proposition 2.1 and Example 2.2 for the negative case and the same theorem and Example 2.3 for the positive one. In particular, the integrability of the case $\epsi=-2/3$ was pointed out in \cite{faruk2} in a different framework.

It is subject of great interest the investigation of integrable equations, although the own concept of integrable equation is not a simple matter. For further discussion about this, see \cite{mik}. In this paper we shall use the following definitions of integrability:

\begin{definition}\label{def1}
An equation is said to be integrable if it admits an infinite hierarchy of higher symmetries.
\end{definition}

An immediate consequence of Theorem \ref{teo4.1} and Definition \ref{def1} is the following:
\begin{corollary}\label{cor1}
Equation $(\ref{1.1})$, with $\epsi a=1$ and $\epsi=\pm 2/3$, is integrable in the sense of Definition $\ref{def1}$.
\end{corollary}

Definition \ref{def1} means integrability in the sense of the symmetry approach. The interested reader is guided to references \cite{miknov,mik,mikso,nov} in which one can find rich material on such subject. Another common definition of integrability is:

\begin{definition}\label{def2}
An equation is said to be completely integrable if it admits an infinite hierarchy of conservation laws.
\end{definition}

One says that an equation $F=0$ is exactly solvable if there exists operators ${\cal L}$ and ${\cal B}$ such that the Lax equation ${\cal L}_{t}=[{\cal B},{\cal L}]$ holds on ${\cal F}$. If there exist such operators, then they are said to form a Lax pair in the sense of Fokas \cite{fok1980} (Lemma 2). Hence, if the equation admits a recursion operator $\mathfrak{R}$, then $\mathfrak{R}$ and $K_{\ast}$ form a Lax pair, see Olver \cite{ol}, Chapter 5. In terms of the equation under consideration in this paper, if $K^\pm[u]=u_{xxx}\mp 3u_xu_{xx}/u$, then it follows that $\mathfrak{R}^\pm_t=[K^\pm_{\ast},\mathfrak{R}^\pm]$ and therefore the operators $\mathfrak{R}^\pm,\,K^\pm_{\ast}$ can be interpreted as Lax pairs of the equation (\ref{1.1}) with $\epsi a=1$ and $\epsi=\pm2/3$, respectively.

In his original paper, however, Lax \cite{lax1968} required a little more from what Fokas called a Lax pair ${\cal L}$ and ${\cal B}$. He assumed that ${\cal L}$ was a self-adjoint operator whose eigenvalues $\lambda$ would not depend on $t$. This would imply the existence of a skew-adjoint operator ${\cal B}$ solving the spectral and temporal problems
\bb\label{spcpr}
{\cal L}\phi=\lambda \phi, \quad \phi_t={\cal B}\phi.
\ee
The Lax equation then arises as compatibility condition of (\ref{spcpr}).

The practical difference between these two viewpoints on Lax pairs is their consequences in terms of Definitions \ref{def1} and \ref{def2}. On the one hand, a recursion operator of (\ref{1.1}) may provide infinitely many symmetries, and therefore, the equation is only integrable \cite{fok1980,mik} in the sense of Definition \ref{def1}. Conversely, a Lax representation proposed by Lax provides infinitely many symmetries as well as infinitely many conservation laws for $(1+1)$-dimensional equations, see \cite{lax1968,mik}, which says the equation is completely integrable.

Therefore, despite the fact that Theorem \ref{teo4.1} is enough to prove integrability in the sense of Definition \ref{def1}, it may be insufficient to prove the complete integrability of equation (\ref{1.1}) with $\epsi = \pm 2/3$.

Let $\phi$ be a smooth function. The integrability of case $\epsi=2/3$ can be assured through the ansätz 
\bb\label{3.2.6}
{\cal L} = -D_x^2+\frac{u_{xx}}{u} \quad \mbox{and} \quad {\cal B} = 4D_x^3-6\frac{u_{xx}}{u}D_x-3\left(\frac{u_{xxx}}{u}-\frac{u_xu_{xx}}{u^2}\right).
\ee
As ${\cal L}$ is a second order differential operator, taking $\phi_1=\phi$ and $\phi_2=\phi_x$, the system (\ref{spcpr}), can be rewritten as 
\[ \ds{\f{\p}{\p x} \left[\begin{array}{c}
\phi_1 \\
\phi_2 \end{array} \right] = U
\left[\begin{array}{c}
\phi_1 \\
\phi_2 \end{array} \right],}\quad\,\,\,
\ds{\f{\p}{\p t} \left[\begin{array}{c}
\phi_1 \\
\phi_2 \end{array} \right] = V
\left[\begin{array}{c}
\phi_1 \\
\phi_2 \end{array} \right]}, \]
where
$$
U = \left[\begin{array}{cc}
0 & 1 \\ \\
\ds{\f{u_{xx}}{u}-\lambda} & 0 \end{array} \right],\quad\quad
V = \left[\begin{array}{cc}
\ds{\frac{u_{xxx}}{u}-\f{u_xu_{xx}}{u^2}} & \ds{2\f{u_{xx}}{u} - 4\lambda} \\ \\
\ds{\left(\f{u_{xx}}{u}\right)_{xx}+6\left(\lambda-\f{u_{xx}}{u}\right)\f{u_{xx}}{u}+4\left(\lambda - \f{u_{xx}}{u}\right)^2}  & \ds{-\frac{u_{xxx}}{u}+\f{u_xu_{xx}}{u^2}} \end{array} \right]
.$$
In terms of a matrix representation, the Lax equation is transformed into a zero-curvature representation
$$\frac{\p U}{\p t}-\f{\p V}{\p x} +[U,V] = 0,$$ which reads
\bb
\left[\begin{array}{cc}
0&0 \\ \\
\ds{\left(\f{u_{xx}}{u}\right)_t-\left(\f{u_{xx}}{u}\right)_{xxx}+6\left(\f{u_{xx}}{u}\right)\left(\f{u_{xx}}{u}\right)_x}&0 \ea \right]=0.
\ee
Observing that
\bb\label{miura}
\left(\f{u_{xx}}{u}\right)_t-\left(\f{u_{xx}}{u}\right)_{xxx}+6\left(\f{u_{xx}}{u}\right)\left(\f{u_{xx}}{u}\right)_x = \left(\f{1}{u}D_x^2-\f{u_{xx}}{u^2}\right)\left(u_t+3\frac{u_xu_{xx}}{u}-u_{xxx}\right),
\ee we conclude that operators (\ref{3.2.6}) provide a Lax representation \cite{lax1968} for equation (\ref{1.1}) with $\epsi a=1$ and $\epsi = 2/3$.

\begin{theorem}\label{teo4.2}
Equation $(\ref{1.1})$, with $\epsi a=1$ and $\epsi=2/3$, admits the Lax pair $(\ref{3.2.6})$ and is completely integrable.
\end{theorem}

\subsection{Higher order conservation laws}

Although the case $\epsi=2/3$ is ``well-behaved'' in the sense it admits a Lax pair and, therefore, infinitely many conservation laws, the ``symmetric'' case $\epsi=-2/3$, on the other hand, at a first sight, seems to be more delicate since we could not find a Lax pair to it. This problem, however, is only apparent. Firstly, we note that this case is the only value of $\epsi$ in (\ref{1.1}) in which the equation admits an infinite dimensional Lie algebra of symmetries\footnote{In Section \ref{notation} as well as in \cite{ijaims} we only considered the finite dimensional case.}. Actually, this case is linearisable, see, for instance, page 34 of \cite{faruk2} or \cite{sen}, and equation (\ref{1.4}) as well.

Another interesting point to be taken into account is that, according to \cite{mikso}, page 75, an equation of the form $u_t=u_{xxx}+F(x,u,u_x,u_{xx})$
is integrable if the following conditions are satisfied:
\bb\label{sys}
\ba{l}
\ds{D_t\left(\f{\p F}{\p u_{xx}}\right)=D_{x}\sigma_1},\,\,\,\,
\ds{D_t\left(3\f{\p F}{\p u_{x}}-\left(\f{\p F}{\p u_{xx}}\right)^2\right)=D_{x}\sigma_2},\\
\\
\ds{D_t\left(9\sigma_1+2\left(\f{\p F}{\p u_{xx}}\right)^3-9\left(\f{\p F}{\p u_{xx}}\right)\left(\f{\p F}{\p u_{x}}\right)+27\f{\p F}{\p u}\right)=D_{x}\sigma_3},\,\,\,\,
\ds{D_{t}\sigma_2=D_{x}\sigma_4.}
\ea
\ee

The reader can check that if one takes $\epsi a=1$ and $F=-2au_xu_{xx}/u$, system (\ref{sys}) is compatible with $\sigma_1=D_t(3\ln{|u|}),\,\,\sigma_2=D_tD_x(9\ln{|u|}),\,\,\sigma_3=D_tD_x^2(27\ln{|u|}),\,\,\sigma_4=D_t^2(9\ln{|u|})$, provided that $\epsi=-2/3$.

The acute reader might have already observed that (\ref{3.1.2}) is not the only possible conserved quantity for equation (\ref{1.1}) with $\epsi a=1$ and $\epsi=-2/3$. Actually, the last line of Table \ref{tab1} shows that 
\bb\label{SIDV}
u_t=u_{xxx}+3\frac{u_xu_{xx}}{u}
\ee
also has the conserved density given by $C^0=u^3Q$ in view of Ibragimov theorem (see \cite{ib2}). The situation becomes much more interesting if one combines this result with the recursion operator, in which one can obtain a hierarchy of conserved quantities given by
$$
H_{n}=\int_{\R}u^3\,\mathfrak{R}^n\,Q\,dx,
$$
where $\mathfrak{R}$ is the recursion operator given by (\ref{3.2.4}) (we omitted the superscript $-$ for simplicity once we believe it would not lead to any confusion at this stage).

Generators $\textbf{v}_1,\,\textbf{v}_2$ and $\textbf{v}_4$ only provide trivial conserved quantities for $n\geq1$. A richer situation arises once the generator $\textbf{v}_3$ is considered. Denoting the conserved densities by
\bb\label{comp}
C^0_n=u^3\,\mathfrak{R}^n\,(3tu_t+xu_x)
\ee
we can easily observe that it will be a nontrivial conserved density (a component of a nontrivial conserved vector) if $\delta C^0_n/\delta u\not\equiv0$, where $\delta/\delta u$ denotes the Euler-Lagrange operator. Below we present a table with some values of $C^0_n$. We opt to only present the component $C^0$ due to the fact that the quantity of terms increases considerably and writing the corresponding fluxes would take a precious and considerable amount of space (and time). Moreover, this also explains why we do not eliminate the term $u_t$: the quantity of terms would be still huge. 
\begin{center}
\begin{longtable}{||l||l||l||l||l|}
\caption{\small{This table shows some higher order conserved densities of equation (\ref{SIDV}). The last column presents the variational derivative of the densities obtained using formula (\ref{comp}). Apart from the case $n=2$, the remaining are non-vanishing results, showing that the corresponding conserved densities are not total derivatives. This means that the conserved vector is not trivial. We adopt the convention $u_{nx}=\p^nu/\p^nx$} and $u_{t,nx}=\p u_{nx}/\p t$.}\\
\hline
 $n$        &   $C^0_n=u^3\,\mathfrak{R}^n\,(3tu_t+xu_x)$ & $\f{\delta C^0_n}{\delta u}$  \\\hline
$1$         &   $\ba{l}2 u^3u_{xx}+3 t u^3u_{txx}+x u^3u_{3x}+2u^2u_x^2\\
+6 t u^2u_xu_{tx} +3 t u^2u_tu_{xx}+3 x u^2u_xu_{xx}\ea$                          & $-4 u \left(u_x^2+u u_{xx}\right)$\\\hline
$2$         &   $\ba{l}4 u^3u_{4x}+3 t u^3u_{t,4x}+x u^3u_{5x}+12 u^2u_{xx}^2\\
+18 t u^2u_{xx}u_{txx}+16 u^2u_xu_{3x} +12 t u^2u_{tx}u_{3x} \\
+10 x u^2u_{xx}u_{3x}+12 t u^2u_xu_{t,3x}+3 t u^2u_tu_{t,4x} \\
+5 x u^2u_xu_{4x}\ea$   & 0 \\\hline
$3$         &      $\ba{l}6 u^3u_{6x}+3 t u^3u_{t,6x}+x u^3u_{7x}+60 u^2u_{3x}^2 \\
+60 t u^2u_{3x}u_{t,3x}+90 u^2u_{xx}u_{4x}+45 t u^2u_{txx}u_{4x}\\
+35 x u^2u_{3x}u_{4x}+45 t u^2u_{xx}u_{t,4x}+36u^2u_xu_{5x}\\
+18 t u^2u_{tx}u_{5x}+21 x u^2u_{xx}u_{5x}\\
+18 t u^2u_xu_{t,5x}+3 t u^2u_tu_{6x}+7 x u^2u_xu_{6x}\ea$
        &  $4 u \left(10 u_{3x}^2+15 u_{xx} u_{4x}+6 u_xu_{5x}+uu_{6x}\right)$\\\hline
$4$         &      $\ba{l}8 u^3u_{8x}+3 t u^3u_{t,8x}+x u^3u_{9x}+280 u^2u_{4x}^2\\
+210 t u^2u_{4x}u_{t,4x}+448 u^2u_{3x}u_{5x}+168 t u^2u_{t,3x}u_{5x}\\
+126 x u^2u_{4x}u_{5x}+168 t u^2u_{3x}u_{t,5x}+224 u^2u_{2x}u_{6x}\\
+84 t u^2u_{txx}u_{6x}+84 x u^2u_{3x}u_{6x}+84 t u^2u_{xx}u_{t,6x}\\
+64 u^2u_xu_{7x}+24 t u^2u_{tx}u_{7x}+36 x u^2u_{xx}u_{7x}\\
+24 t u^2u_xu_{t,7x}+3 t u^2u_tu_{8x}+9 x u^2u_xu_{8x}\ea$
        &  $\ba{l}8 u \left(35 u_{4x}^2+56 u_{3x}u_{5x}+28 u_{xx}u_{6x}\right)\\
        +8u(8 u_xu_{7x}+uu_{8x})\ea$\\\hline

\end{longtable}
\end{center}


Table 2 shows higher order conserved densities derived from symmetries. This suggests that once an equation (mainly the evolutionary ones) has a recursion operator and point symmetries known, one can try to obtain higher order conservation laws using Ibragimov theorem. This might be used as an integrability test. On the other hand, the techniques \cite{baprl,abeu1,abeu2} can equally be employed for the same purpose. In \cite{ancoib} the reader can find applications in this direction regarding Krichever-Novikov type equations. For further discussion, see \cite{anco}.

\section{Miura type transformations}\label{secmiu}
In his celebrated paper \cite{miura}, Miura exhibited a nonlinear transformation mapping solutions of the mKdV equation
\bb\label{m1}
v_t=v_{xxx}-6v^2v_x
\ee
into solutions of the KdV equation
\bb\label{m2}
w_t=w_{xxx}-6ww_x.
\ee
The mentioned transformation, called Miura's transformation, can be written as $w=v^2+v_x$. The reader has probably noted a similarity between the left hand of equation (\ref{miura}) and KdV equation (\ref{m2}). Actually, let $u$ be a solution of (\ref{1.1}) with the constraints $\epsi a=1$ and $\epsi=2/3$. A subtle consequence of equation (\ref{miura}) beyond its core is the fact that if we define 
\bb\label{m3}
w=\f{u_{xx}}{u},
\ee
then $w$ is a solution of the KdV equation! From Figure \ref{fig2} we have the following sequence of transformations:

\begin{figure}[h!]
\centering
\xymatrix{ \ds{v_t=v_{xxx}-6v^2v_x} \quad\ar@/^1.5cm/[rrrr]^{\quad\ds{w=\f{u_{xx}}{u}=\f{D_{x}^2(e^{D_x^{-1}v})}{e^{D_x^{-1}v}}=v^2+v_x}}\ar[rr]^{\ds{u=e^{D_x^{-1}v}}}&  &\quad \ds{u_t+3\frac{u_xu_{xx}}{u}=u_{xxx}}\quad \ar[rr]^{\quad\quad \ds{w=\f{u_{xx}}{u}}}& &\quad \ds{w_t=w_{xxx}-6ww_x} }
\caption{\small{The diagram shows the sequence of nonlinear transformations mapping solutions of the mKdV into solutions of equation (\ref{1.1}) with $\epsi a=1$ and $\epsi=2/3$ and next, solutions of the latter into the KdV equation. The composition of these two transformations is just Miura's transformation.}}\label{fig2}
\end{figure}

Transformation (\ref{m3}) not only provides solutions to the KdV equation from known solutions of 
\bb\label{sidv+}
u_t+3\frac{u_xu_{xx}}{u}=u_{xxx},
\ee
but it also enables us to obtain solutions of (\ref{sidv+}) once a solution of (\ref{m2}) is given through the equation
\bb\label{sch}
{\cal L}_w\,u=0,
\ee 
where $\w$ is the Schrödinger operator
\bb\label{ope}
\w=\f{\p^2}{\p x^2}-w
\ee 
and the potential $w$ is a solution of (\ref{m2}). If we denote the set of classical solutions of (\ref{m2}) by ${\cal S}$, we have a mapping ${\cal S}\ni w\mapsto \w$.

\begin{rem}
We restrict ourselves to define ${\cal S}$ as the set of classical solutions of the KdV equation $(\ref{m2})$ for convenience in order to expose our ideas and avoid technicalities.
\end{rem}

In what follows, ${\cal N}(\w)$ denotes the kernel of the operator (\ref{ope}).
\begin{theorem}\label{teo4.3}
Let ${\cal S}$ be the set of classical solutions of the KdV equation $(\ref{m2})$ and $w\in{\cal S}$. If $u\in{\cal N}(\w)$ is a non identically vanishing solution of
\bb\label{sidvlin}
u_t+3wu_x=u_{xxx},
\ee 
then $u$ is a solution of the equation $(\ref{sidv+})$.
\end{theorem}
\begin{proof}
The proof follows from (\ref{m3}) and (\ref{sidvlin}).
\end{proof}

\begin{corollary}\label{cor4.2}
Let $\al\in\R$ and $w\in{\cal S}$. If $u^1,\,u^2\in{\cal N}(\w)$ are non identically vanishing solutions of $(\ref{sidvlin})$, then $u^1+\al u^2$ is a solution of $(\ref{sidv+})$.
\end{corollary}

\begin{proof}
From Theorem \ref{teo4.3}, we have
$$
\ba{l}
\ds{u^1_t+3\frac{u^1_xu^1_{xx}}{u^1}=u^1_{xxx}},\quad
\ds{u^2_t+3\frac{u^2_xu^2_{xx}}{u^2}=u^2_{xxx}}.
\ea
$$
Multiplying the last equation by $\al$ and summing, we have
\bb\label{4.45}
(u^1+\al u^2)_t+3\left(\frac{u^1_xu^1_{xx}}{u^1}+\al\frac{u^2_xu^2_{xx}}{u^2}\right)=(u^1+\al u^2)_{xxx}.
\ee

On the other hand, since $u^1,\,u^2\in{\cal N}(\w)$, they satisfy (\ref{m3}). Therefore,
\bb\label{4.46}
\frac{u^1_xu^1_{xx}}{u^1}+\al\frac{u^2_xu^2_{xx}}{u^2}=w\,u^1_x+\al w\,u^2_x=w(u^1+\al\,u^2)_x.
\ee

Substituting (\ref{4.46}) into (\ref{4.45}) we conclude that $
(u^1+\al u^2)_t+3w(u^1+\al\,u^2)_x=(u^1+\al u^2)_{xxx}$. By Theorem \ref{teo4.3}, $u^1+\al u^2$ is a solution of (\ref{sidv+}).
\end{proof}

\begin{rem}
Theorem $\ref{teo4.3}$ says that, for each $w\in{\cal S}$, a solution of $(\ref{sidv+})$ can be obtained by solving the linear system
\bb\label{4.47}
\left\{
\ba{l}
u_t+3wu_x=u_{xxx},\\
\\
u_{xx}-wu=0.
\ea
\right.
\ee
\end{rem}

Here we present some examples illustrating how Theorem \ref{teo4.3} is useful for finding solutions of (\ref{sidv+}).

\begin{exe}\label{ex4.1}
Let us consider $w=1$. The solution of the second equation of $(\ref{4.47})$ is given by
$u(x,t)=a(t)e^{x}+b(t)e^{-x}$. Substituting this function into the first equation of $(\ref{4.47})$ we conclude that 
$u(x,t)=c_1\,e^{x-2t}+c_1\,e^{-x+2t}$.

Note that the last function provides an infinite number of solutions of $(\ref{4.47})$ likewise linear equations, as stated in Corollary $\ref{cor4.2}$. On the other hand, if we consider $w=-1$, proceeding as before, we obtain the solution
$u(x,t)=c_1\,\cos{(x+2t)}+c_2\,\sin{(x+2t)}.$
\end{exe}

\begin{exe}\label{ex4.2}
A simple inspection shows that $w(x,t)=x/(6t)$ is a solution of $(\ref{m2})$. For convenience, let us assume that $t>0$. The last equation of $(\ref{4.47})$ is an Airy equation\footnote{Do not make confusion between the Airy equation $(\ref{1.4})$ with the Airy ordinary differential equation $y''(z)-k^2zy(z)=0$.}, whose solution is given by
\bb\label{4.50}
u(x,t)=a(t)\,\Ai\left(\f{x}{(6t)^{1/3}}\right)+b(t)\,\Bi\left(\f{x}{(6t)^{1/3}}\right),
\ee
where $\Ai(\cdot)$ and $\Bi(\cdot)$ are the Airy functions, see \cite{wolfram}. Substituting $(\ref{4.50})$ into the first equation of $(\ref{4.47})$, we have
\bb\label{4.51}
u(x,t)=c_1\,t^{1/6}\,\Ai\left(\f{x}{(6t)^{1/3}}\right)+c_2\,t^{1/6}\,\Bi\left(\f{x}{(6t)^{1/3}}\right),
\ee
where $c_1$ and $c_2$ are arbitrary constants.

\begin{center}
\begin{figure}[h!]

\centering
    \begin{subfigure}[b]{0.35\textwidth}
           \includegraphics[width=\textwidth]{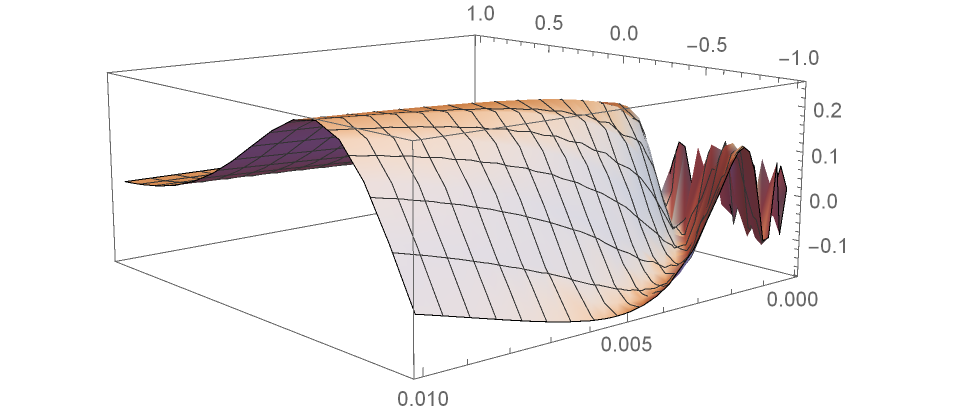}
        \caption{}
        \label{fig:a}
    \end{subfigure}
    \begin{subfigure}[b]{0.35\textwidth}
           \includegraphics[width=\textwidth]{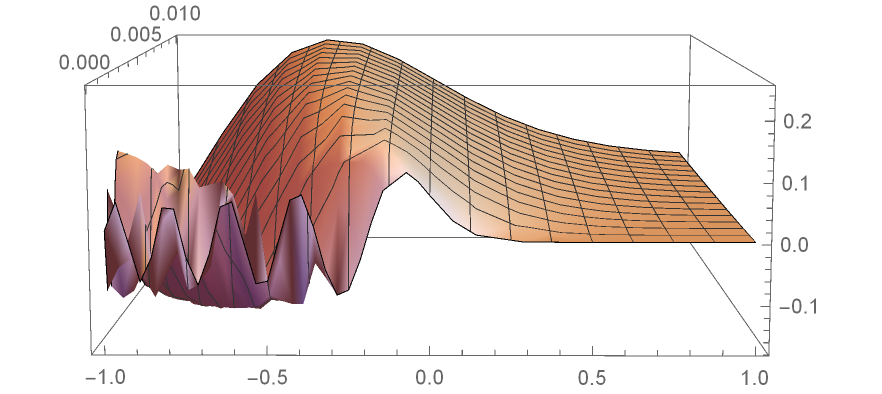}
        \caption{}
        \label{fig:b}
    \end{subfigure}
        \begin{subfigure}[b]{0.25\textwidth}
           \includegraphics[width=\textwidth]{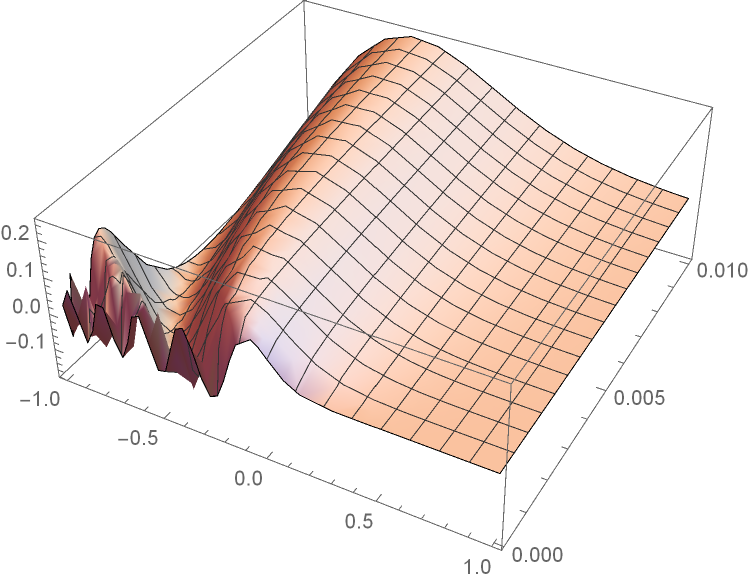}
        \caption{}
        \label{fig:c}
    \end{subfigure}
           {\small\caption{The figure shows different perspectives of the solution (\ref{4.51}), with $c_1=1$ and $c_2=0$, for $x\in[-1,1]$ and $t\in[0.0001,0.01]$.}}
\end{figure}
\end{center}

\begin{center}
\begin{figure}[h!]
\centering
    \begin{subfigure}[b]{0.3\textwidth}
           \includegraphics[width=\textwidth]{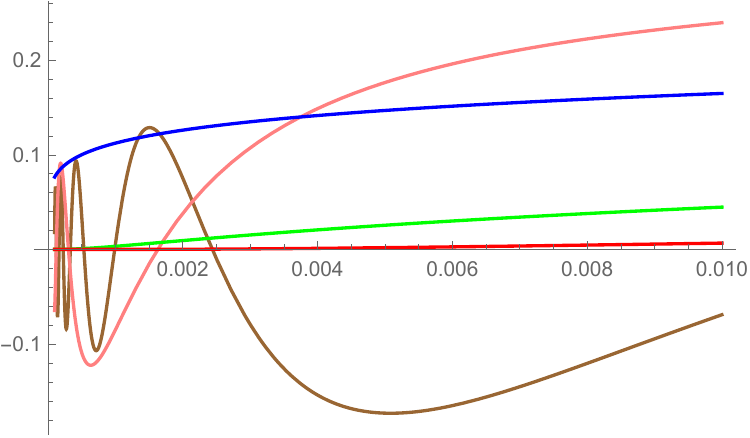}
        \caption{}
        \label{fig:4a}
    \end{subfigure}
    \begin{subfigure}[b]{0.3\textwidth}
           \includegraphics[width=\textwidth]{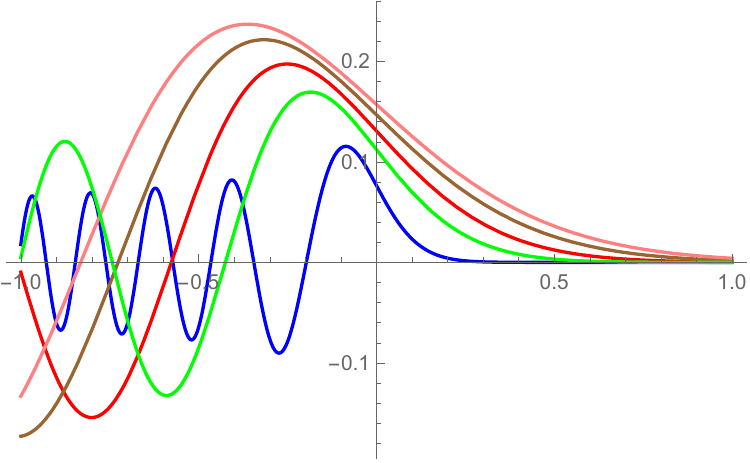}
        \caption{}
        \label{fig:4b}
    \end{subfigure}
                 {\small\caption{Graphics of $u(x,t)=t^{1/6}\Ai(x/(6t)^{1/3})$ for fixed arguments. Figure $(a)$ shows $u(x_0,t)$ for $t\in[0.0001,0.01]$ and $x_0=$$-1$ (brown),$-1/2$ (pink), $0$ (blue), $1/2$ (green), and $1$ (red). Figure $(b)$ shows $u(x,t_0)$ for $x\in[-1,1]$ and $t_0=$$0.0001$ (blue), $0.0025$ (red), $0.005$ (brown), $0.0075$ (pink), and $0.001$ (green).}}
\end{figure}
\end{center}

\end{exe}

\section{Solitary waves}\label{sol}

In \cite{ijaims} it was proved that (\ref{1.7}) is a solution of (\ref{1.1}) provided that $k^2<0$, where $k^2$ is that in (\ref{1.6}). Otherwise, if $k^2>0$, one would have an oscillatory solution, given by
$
u(x,t)=\sin{\left(\sqrt{c/a(\epsi-2)}(x-ct)+\phi_0\right)}$, where $\phi_0$ is an initial phase.

A simple inspection shows that for $\epsi=2$, then $u(x,t)=e^{\pm kx}$ is a solution of (\ref{1.1}) as well as $u(x,t)=\sin{(kx+\phi_0)}$, where in both cases $k$ is an arbitrary real number and $\phi_0$ is an initial phase.

The sinusoidal solution presented before is a wave solution, but it is not a solitary wave, that is, a solution of the type $u=\phi(x-ct)$ such that $u\rightarrow u_{\pm}$ whenever $x-ct\rightarrow\pm\infty$ and $u_\pm$ are constants.

In this section we shall investigate the existence of solitary waves other than the sech-squared presented at the very beginning of the paper and motivated the discovery of equation (\ref{1.1}).

We firstly begin with Theorem \ref{teo4.3} and the 1-soliton solution
\bb\label{6.1}
w(x,t)=-\f{c}{2}\sech^2{\left(\f{\sqrt{c}}{2}(x+ct)\right)}
\ee
of the KdV equation\footnote{Note that the change $(x,t,w)\mapsto(x,-t,-w)$ transforms equation (\ref{m2}) into $w_t+w_{xxx}+6ww_x=0$. This is equivalent to map (\ref{6.1}) into the $\sech^2$ solution showed at the very beginning of the paper, which is a solution of $w_t+w_{xxx}+6ww_x=0$.} (\ref{m2}).

The second type of solitary wave we shall consider is the weak one, in the distributional sense. More precisely, in view of the solutions (\ref{1.7}), we are tempted to determine whether equation (\ref{1.1}) admits certain weak solitary waves as solutions. To look for these solutions, it is more convenient to rewrite (\ref{1.1}) as
\bb\label{4.1}
c\phi'-2a\frac{\phi'\phi''}{\phi}+\epsi a \phi'''=0,
\ee
upon the change $u=\phi(z),\,\,z=x-ct$. In (\ref{4.1}) the prime `` $'$ '' means derivative with respect to $z$. Other solutions can be found in \cite{sen,ijaims}. 

\subsection{Classical solitary waves}\label{kink}
As pointed out before, we make use of Theorem \ref{teo4.3} for finding solutions to (\ref{1.1}). Actually, the theorem itself imposes a restriction on the parameters of (\ref{1.1}): $\epsi a=1$ and $a=3/2$, which correspond to the completely integrable equation (\ref{sidv+}).

Substituting (\ref{6.1}) into the second equation of (\ref{4.47}) and following the trick suggested in  \cite{ablo}, page 257, exercise 9.11, or page 46 of \cite{drazin}, defining $z=\sqrt{c}/2(x+ct)$, $\xi=\tanh{z}$ and $v(z)=\psi(\xi)$, we transform
\bb\label{6.1.1}
u_{xx}+\f{c}{2}\sech^2{\left(\f{\sqrt{c}}{2}(x+ct)\right)}u=0
\ee
into
\bb\label{6.1.2}
(1-\xi^2)\psi''-2\xi \psi'+2\psi=0.
\ee

Equation (\ref{6.1.2}) is the associated Legendre equation, whose solution is
$$\psi(\xi)=c_1\xi+c_2\left(\f{\xi}{2}\ln{\left(\f{1-\xi}{1+\xi}\right)}+1\right),
$$
which has the solution $u(x,t)=u_1 + u_2$, where 
$$
u_1=A(t)\tanh{\left(\f{\sqrt{c}}{2}(x+ct)\right)},\quad u_2=B(t)\left(1-\f{\sqrt{c}}{2}(x+ct)\tanh{\left(\f{\sqrt{c}}{2}(x+ct)\right)}\right),
$$
$A$ and $B$ are arbitrary smooth functions of $t$, which corresponds to two linearly independent solutions of (\ref{6.1.1}). Substituting the foregoing functions into the first equation of (\ref{4.47}), we conclude that $A(t)=1$ and $B(t)=0$. Therefore, by Theorem \ref{teo4.3}
\bb\label{6.1.4}
u(x,t)=\tanh{\left(\f{\sqrt{c}}{2}(x+ct)\right)}
\ee
is a solution to equation (\ref{sidv+}).
\begin{figure}[h!]
        \setbox0\hbox{%
                \includegraphics[width=.4\textwidth]{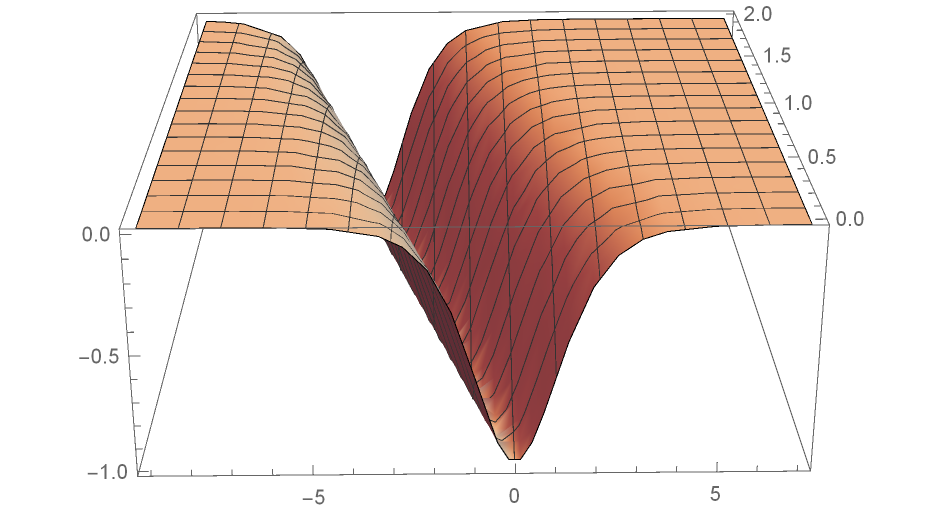}%
        }%
        \setbox2\hbox{%
                \includegraphics[width=.4\textwidth]{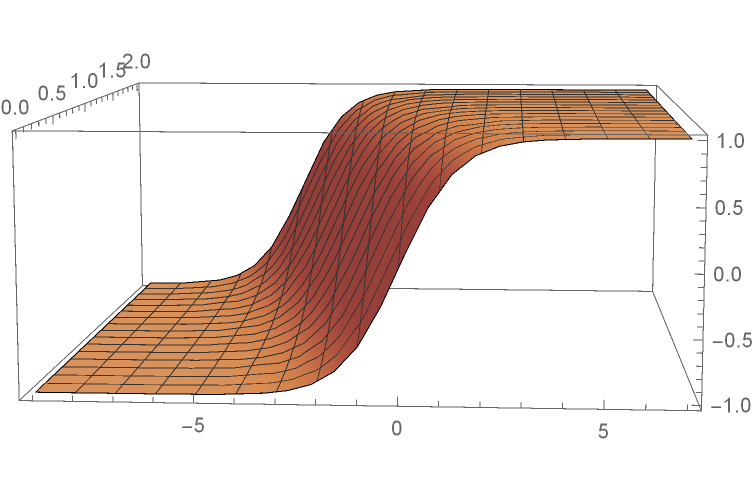}%
        }%
        \ifdim\ht0>\ht2
                \setbox0\hbox{%
                        \includegraphics[height=\ht2]{sech_2-c_2-_x_-9_7_-_t_0_2_}%
                }%
        \else
                \setbox2\hbox{%
                        \includegraphics[height=\ht0]{tanh-c_2-_x_-9_7_-_t_0_2_}\label{fig:5b}
                }%
        \fi
        \noindent
        \parbox{.4\textwidth}{%
                \centering
                \unhbox0
                \caption*{(a)}
                \label{fig:5a}
        }%
        \hfil
        \parbox{.4\textwidth}{%
                \centering
                \unhbox2
                \caption*{(b)}
                \label{fig:5b}
        }%
        \caption{{\small Figure (a) shows the solution (\ref{6.1}) of (\ref{m2}), whereas Figure (b) exhibits solution (\ref{6.1.4}) of (\ref{sidv+}). In both cases it is chosen $c=2$ and $x\in[-9,7]$ and $t\in[0,2]$.}}\label{fig5}
\end{figure}

Solution (\ref{6.1.4}) is a bounded and monotonic solitary wave (see Figure \ref{fig5}), that is, a kink solution. Therefore, the 1-soliton solution (\ref{6.1}) of (\ref{sidv+}) is ``transformed'', through Theorem \ref{teo4.3}, into the kink solution (\ref{6.1.4}). Vice-versa, the solution $u(x,t)$ of the equation (\ref{sidv+}), given by (\ref{6.1.4}), is mapped as $w(x,t)$, given by equation (\ref{6.1}), through transformation (\ref{m3}).

The reader might be thinking about the meaning of the solution $u_2$ above since $B(t)=0$. Clearly it is not a solution to (\ref{sidv+}) and this may suggest a contradiction with the fact that solutions $u$ of (\ref{4.47}) are solutions of (\ref{sidv+}). The incongruence is only apparent: in conformity with Theorem \ref{teo4.3}, only the non vanishing solutions $u$ of the system (\ref{4.47}) are solutions of (\ref{sidv+}).

\subsection{Weak solitary waves}

Peakon solutions were introduced in the prestigious work of Camassa and Holm \cite{camassa} and can be described as the follows: a peakon is a wave with a pointed crest at which there are the lateral derivatives, both finite but not equal. We shall firstly promote a na\"ive discussion on peakons and next prove the existence of such solutions.

\subsubsection{Peakons: a heuristic discussion}

Let $I\subseteq\R$ be an interval and suppose that a function $\phi$ is continuous on it. One says that $\phi$ has a {\it peak} at a point $x\in I$ if $\phi$ is smooth on both $I\cap\{z\in\R;\,z<x\}$ and $I\cap\{z\in\R;\,z>x\}$ and
$$0\neq\lim_{\epsilon\rightarrow 0^+}\phi'(x+\epsilon)=-\lim_{\epsilon\rightarrow 0^+}\phi'(x-\epsilon)\neq\pm\infty.$$

Then one says that a function $\phi$ is a peakon solution of (\ref{4.1}) if it is a solution of (\ref{4.1}) having a peak. Such solution, if it exists, should be considered in the distributional sense. For further details, see \cite{lenells09,lenells}.

\begin{figure}[h!]
        \setbox0\hbox{%
                \includegraphics[width=.4\textwidth]{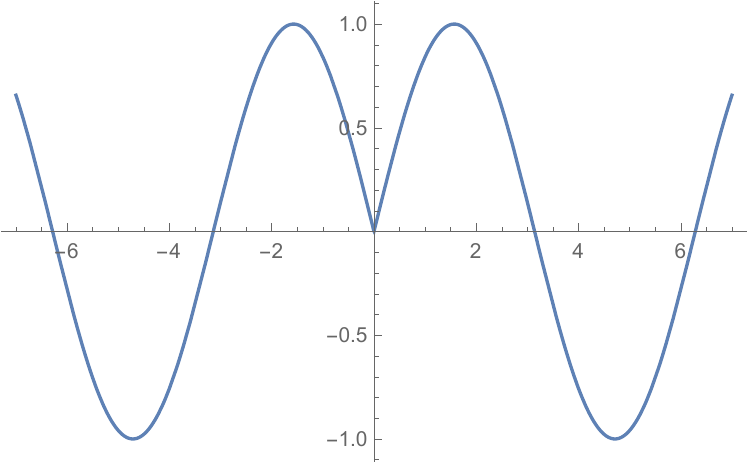}%
        }%
        \setbox2\hbox{%
                \includegraphics[width=.4\textwidth]{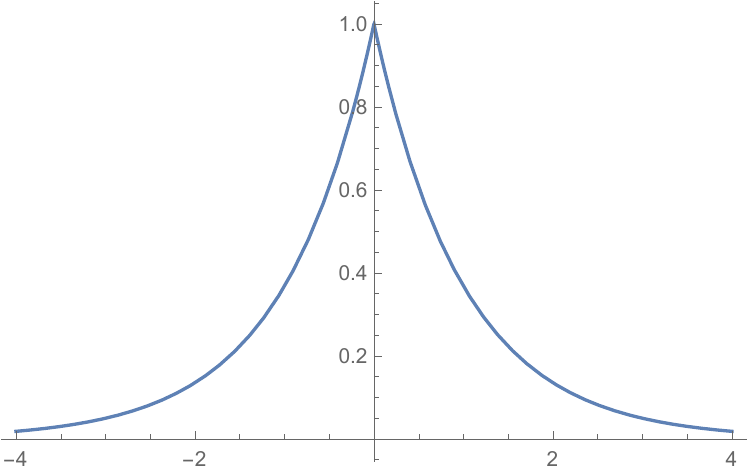}%
        }%
        \ifdim\ht0>\ht2
                \setbox0\hbox{%
                        \includegraphics[height=\ht2]{Untitled-1.pdf}%
                }%
        \else
                \setbox2\hbox{%
                        \includegraphics[height=\ht0]{Untitled-2.pdf}\label{fig:5b}
                }%
        \fi
        \noindent
        \parbox{.4\textwidth}{%
                \centering
                \unhbox0
                \caption*{(a)}
                \label{fig:5a}
        }%
        \hfil
        \parbox{.4\textwidth}{%
                \centering
                \unhbox2
                \caption*{(b)}
                \label{fig:5b}
        }%
        \caption{Example of functions having a peak at the origin: (a) $\phi(z)=\sin|z|$ and (b) $\phi(z)=e^{-|z|}$.}\label{2d}
\end{figure}

The arguments presented before are enough to expose peakons of real valued functions to the reader, but they might not be satisfactorily acceptable for a function $u=u(x,t)$. We can easily make a natural extension as the follows: a continuous function $u(x,t)$ is said to have a peakon at a point $(x_0,t_0)$ if at least one of the functions $x\mapsto u(x,t_0)$ or $t\mapsto u(x_0,t)$ has a peak at $x_0$ or $t_0$, respectively.

We believe that peakons have adequately been discussed for what we need of them. Thus, moving forward, our first step in this section is deciding for which parameters equation (\ref{1.1}) might admit a peakon solution of the form $u(x,t) = e^{-A|x-ct|},$ where $A$ is a constant to be determined.

Consider the function, 
\bb\label{4.2}
\phi(z) = e^{-A|z|},\,\,A=const.
\ee
One would observe that $\phi\in L^1_{loc}(\R)$. Furthermore, if $n$ is a positive integer, then we observe that $\phi(n z)=\phi(z)^n$, which implies $\phi(z)^n$ is also a distribution for any positive integer $n$. 

Originally, we had thought that equation (\ref{1.1}) would have peakon solutions for any values of $\epsi$, see \cite{ijcnmac}. However, a little later we surprisingly found out that it was not the case: peakons would only be admitted if $\epsi=0$. This is a very sensitive point and we would therefore spend some time trying to explain the reason. The heuristic discussion we will present now, to be formally proved in next subsection, is related to that one presented by Lenells \cite{lenells09}, in which he beautifully explained the formation of peakons in the Camassa-Holm equation and why this solution could not be admitted by another evolution equation studied in the same paper.

Going back to equation (\ref{4.1}), it can be alternatively written in terms of a travelling wave $u=\phi(z)$ as
\bb\label{4.8}
\phi'\left(\f{c}{2a}\phi-\phi''\right)+\f{\epsi}{2} \phi\phi'''=0.
\ee

Considering $\phi$ as in (\ref{4.2}) and noticing that $\f{c}{2a}\phi-\phi'' = \left(\f{c}{2a}-A^2\right)\phi+2A\delta (z)$
 in the weak sense and $\phi'$ is the product of the signal function by another one, defined on $0$, then, from the regularisation $\sign{(0)}=0$, the first term in (\ref{4.8}) vanishes just by choosing $A^2=c/2a$. The remaining term does not vanishes identically, unless $\epsi=0$. For further discussion, see \cite{lenells09,iscripta}.

This implies that $\epsi=0$ leads to the peakon solutions
\bb\label{4.6}
\ds{u(x,t)=e^{-\sqrt{\f{c}{2a}}|x-ct|}}
\ee
or
\bb\label{4.7}
u(x,t)=\sin\left(\sqrt{-\frac{c}{2a}}|x-ct|\right)
\ee
with $c \neq 0$, depending on whether \text{$\sign{(c)}=\sign{(a)}$} or $\sign{(c)}=-\sign{(a)}$, respectively.

\begin{rem}\label{remX}
The function $(\ref{4.7})$ has a peak at $x=ct$, but it does not satisfy $u(x,t)\rightarrow u_\pm$, $u_\pm$ constants, when $x-ct\rightarrow\pm\infty$. This comes from the fact that $u(x,t)=\phi(z)$, where $\phi$ is given by $(\ref{4.2})$, but $A$ is purely imaginary, providing an oscillatory and non vanishing function.
\end{rem}

\subsubsection{Existence of a global weak solitary wave to (\ref{1.1})}

In this subsection we show in a rigorous way in what sense and restrictions equation (\ref{1.1}) admits weak solitary waves as solutions.

First and foremost, peakons are solutions in the distributional sense. Thus, if $u=u(x,t)$ is a distributional solution to (\ref{1.1}), all terms in (\ref{1.1}) must be locally integrable on a certain domain, in which the solution is defined on. Physically thinking, one would take $x$ as the space variable and $t$ as the time. Then a natural choice would be $x\in\R$ and $t\in[0,T)$, for a certain $T>0$. 

The first integral (\ref{3.1.3}) gives us an insight on where the solution lies: it should belong to $L^2(\R)$. However, this information is not enough, once we do not have any further information in which space its derivatives are. One should then request that its derivatives up to third order with respect to $x$ are locally integrable, that is, $u(\cdot,t)\in W^{3,1}_{loc}(\R)$ and $u(x,\cdot)$ is a distribution in $L_{loc}^{1}[0,T)$.

However, this might not be enough to assure that the term $u_xu_{xx}/u$ is well-behaved. This problem can be evaded by requesting that  $u\not\equiv0$  and $u_xu_{xx}/u\in L^1_{loc}(\R)$, although we do not avoid the existence of points $(x_0,t_0)\in\R\times[0,T)$ in which $u(x,t)=0$. They can exist, but if they exist, they must be such that $u_xu_{xx}/u$ is bounded near $(x_0,t_0)$. 

These observations enable us to put our problem as the follows: one needs to determine whether the problem
\bb\label{5.2.1}
\left\{\ba{l}
\ds{u_t+\frac{2a}{u}u_xu_{xx}=\epsilon au_{xxx}},\quad x\in\R,\,\,t\in[0,T),\\
\\
u(x,0)=u_0(x),\quad x\in\R,\\
\\
u_x(\cdot,t)u_{xx}(\cdot,t)/u(\cdot,t)\in L^1_{loc}(\R),\\
\\
0\not\equiv u\rightarrow 0\quad\text{as}\quad |x|\rightarrow\infty,
\ea\right.
\ee
admits peakon solutions.

Making the weak formulation of (\ref{1.1}), one has the following definition.

\begin{definition}\label{defsol}
Given an initial data $u_0\in W^{3,1}(\R)$, a function $u\in L^\infty_{loc}(W^{3,1}(\R),[0,T))$ is said to be a weak solution to the initial-value problem $(\ref{5.2.1})$ if it satisfies the identity
\bb\label{5.2.2}
\int_{\R}u_{0}(x)\varphi(x,0)dx+\int_{0}^T\int_{\R}\left(u\varphi_t-2a\f{u_xu_{xx}}{u}\varphi+\epsi a u\varphi_{xxx}\right)dx\,dt=0,
\ee
for any smooth test function $\varphi\in C^\infty_0(\R\times[0,T))$. If $u$ is a weak solution on $[0,T)$ for every $T>0$, then $u$ is called a global weak solution.
\end{definition}

\begin{theorem}\label{teo5.1}
Assume that $a>0$. For any $c>0$, the peaked function $(\ref{4.6})$ is a global weak solution to $(\ref{5.2.1})$ in the sense of Definition $\ref{defsol}$ if and only if $\epsi=0$.
\end{theorem}

\begin{rem}\label{rem1}
If $a<0$ we still have a peakon solution to $(\ref{1.1})$. In this case, one should change $c$ by $-c$ in $(\ref{4.6})$.
\end{rem}

\begin{rem}\label{rem2}
The sinusoidal solution $(\ref{4.7})$ is a wave solution to $(\ref{1.1})$ satisfying all but the last condition of $(\ref{5.2.1})$.
\end{rem}

\begin{rem}\label{rem2}
Since $a\neq0$, it is enough to prove that
\bb\label{5.2.3}
\ds{u(x,t)=e^{-\sqrt{\f{c}{2}}|x-ct|}}
\ee
is a solution to the problem $(\ref{5.2.1})$ with $a=1$. 
\end{rem}

We shall omit the proof of the next three lemmas.

\begin{lemma}\label{lema5.1}
Let $u:\R\times[0,T)\rightarrow\R$ be the function given by $(\ref{5.2.3})$. Then, for all $x\in\R$ and $t\in[0,T)$, $u_x=-\sqrt{c/2}\,\sign{(x-ct)}\,u(x,t)$, $u_t=c\,\sqrt{c/2}\,\sign{(x-ct)}\,u(x,t)$, $u_{xx}=-2\,\delta(x-ct)+c\,u/2$ in the distributional sense.
\end{lemma}

\begin{lemma}\label{lema5.2}
Let $u$ be the function $(\ref{5.2.3})$. Then $u_xu_{xx}$ and $u_xu_{xx}/u$ are well defined and belong to $L_{loc}^1(\R)$. In particular, $u_xu_{xx}=-(c/2)^{3/2}\sign{(x-ct)}u^2$.
\end{lemma}

\begin{lemma}\label{lema5.3}
Let $u$ be the function $(\ref{5.2.3})$ and $u_0(x):=u(x,0),\,\,x\in\R$. Then 
$$
\lim_{t\rightarrow0^+}\|u(x,t)-u_0(x)\|_{W^{1,\infty}}=0.
$$
\end{lemma}

{\bf Proof of Theorem \ref{teo5.1}}: We first prove that (\ref{5.2.3}) solves (\ref{5.2.2}) with $a=1$ for an arbitrary $T>0$. Then the result follows from the arbitrariness of $T$.

By Lemma \ref{lema5.3}, we conclude that $u$ satisfies the initial condition given in (\ref{5.2.1}) (taking $u$ as in (\ref{5.2.3}) and $u_0(x)=u(x,0)$) and, from Lema \ref{lema5.2}, the third condition is satisfied. Clearly the fourth condition is also accomplished. Therefore, the only condition we really need to check is if (\ref{5.2.3}) is a weak solution of the first equation in (\ref{5.2.3}).

Let $\varphi\in C^\infty_0(\R\times[0,T))$ and
$$
\ba{l}
I_1=\ds{\int_{\R}u_{0}(x)\varphi(x,0)dx+\int_{0}^T\int_{\R}u\varphi_tdx\,dt},\\
\\
I_2:=\ds{2\int_{0}^T\int_{\R}\f{u_xu_{xx}}{u}\varphi dx\,dt},\\
\\
I_3:=\ds{\epsi\int_{0}^T\int_{\R} u\varphi_{xxx}dx\,dt.}
\ea
$$
By Fubini's theorem and Lemma \ref{lema5.1}, we have
\bb\label{5.2.4}
\ba{l}
I_1=\ds{\int_{\R}u_{0}(x)\varphi(x,0)dx+\int_\R\int_{0}^Tu\varphi_t dx dt=\int_Ru_{0}(x)\varphi(x,0)dx}\\
\\
\ds{+\int_\R\left(\left.u(x,t)\varphi(x,t)\right|_{0}^T -\int^T_0c\sqrt{\f{c}{2}}\sign{(x-ct)}e^{-\sqrt{\f{c}{2}}|x-ct|}\varphi(x,t)dt\right) dx}\\
\\
\ds{
=-c\sqrt{\f{c}{2}}\int^T_0\int_\R\sign{(x-ct)}e^{-\sqrt{\f{c}{2}}|x-ct|}\varphi(x,t)dt dx.}
\ea
\ee

By Lemma \ref{lema5.2}
\bb\label{5.2.5}
I_2=-2\left(\f{c}{2}\right)^{3/2}\int_{0}^T\int_{\R}\sign{(x-ct)}e^{-\sqrt{\f{c}{2}}|x-ct|}\varphi dx\,dt
\ee
and, by Lemma \ref{lema5.1} and the properties of weak derivatives,
\bb\label{5.2.6}
\ba{l}
I_3=\ds{\epsi\int_{0}^T\int_{\R} \left(-2\delta (x-ct)+\f{c}{2}e^{-\sqrt{\f{c}{2}}|x-ct|}\right)\varphi_{x}dx\,dt}\\
\\
\ds{-\epsi\int_{0}^T\int_{\R} \left(-2\delta' (x-ct)-\left(\f{c}{2}\right)^{3/2}\sign{(x-ct)}e^{-\sqrt{\f{c}{2}}|x-ct|}\right)\varphi dx\,dt}\\
\\
=\ds{\epsi \varphi_x(ct,t)+\epsi\left(\f{c}{2}\right)^{3/2}\int_{0}^T\int_{\R}\sign{(x-ct)}e^{-\sqrt{\f{c}{2}}|x-ct|}\varphi dx\,dt}.
\ea
\ee

On the other hand, for all $\varphi\in C^\infty_0(\R\times[0,T))$, we have
\bb\label{5.2.7}
\ba{l}
\ds{\epsi\left(\varphi_x(ct,t)+\left(\f{c}{2}\right)^{3/2}\int_{0}^T\int_{\R}\sign{(x-ct)}e^{-\sqrt{\f{c}{2}}|x-ct|}\varphi dx\,dt\right)}=I_1-I_2+I_3\\
\\
\ds{=\int_{\R}u_{0}(x)\varphi(x,0)dx+\int_{0}^T\int_{\R}\left(u\varphi_t-2\f{u_xu_{xx}}{u}\varphi+\epsi  u\varphi_{xxx}\right)dx\,dt}.
\ea
\ee

If $\epsi=0$, then $u$ satisfies (\ref{5.2.2}). On the other hand, if (\ref{5.2.2}) is holds, then we must have  
$$
\ds{\epsi\left(\varphi_x(ct,t)+\left(\f{c}{2}\right)^{3/2}\int_{0}^T\int_{\R}\sign{(x-ct)}e^{-\sqrt{\f{c}{2}}|x-ct|}\varphi dx\,dt\right)}=0
$$
for all $\varphi\in C^\infty_0(\R\times[0,T))$, which implies $\epsi=0$.

Therefore, if $\epsi=0$, then (\ref{5.2.2}) holds for every $T>0$ and (\ref{4.6}) is a global weak solution to (\ref{5.2.1}).\quad$\square$

\begin{rem}
Note that we only need $u\in L^{\infty}_{loc}(W^{2,1}(\R),[0,T))$ since $\epsi=0$.
\end{rem}

\section{Discussion}\label{discussion}

Equation (\ref{1.1}) was introduced as a generalisation of an equation obtained by a program, see \cite{sen} and references thereof. There the authors explored several properties of (\ref{1.1}) and also raised quite intriguing questions on its properties, as we pointed out in the introduction of this paper.

We confirmed that the values observed in \cite{sen} are, indeed, special cases as they also appear in the Lie symmetry approach as exceptional values and in the nonlinear self-adjoint classification as well. In virtue of the results proved in \cite{ijcnsns,rita}, the approach used in \cite{ib2,ib6} was a natural choice for establishing local conservation laws for (\ref{1.1}). Although the techniques \cite{baprl,abeu1,abeu2,2ndbook} lead to the same results, it should also be noted that the techniques introduced in \cite{ib2} may provide an integrability test once symmetries and recursion operators are known as suggested in section \ref{cl}.

An important question raised in \cite{sen} is if, for some $\epsi$, equation (\ref{1.1}) could be transformed into the KdV equation. We answered this question positively: actually, the response is just equation (\ref{sidv+}). Moreover, we also exhibited a Lax pair for it. These two compatible operators are a threefold discovery: first and foremost, the Lax pair has interest by itself, assuring the integrability of the equation. Secondly, it is a cornerstone to prove Theorem \ref{teo4.3} and Corollary \ref{cor4.2}, from which we can obtain solutions of equation (\ref{sidv+}) from the solutions of the KdV equation (\ref{m2}) by solving the linear system (\ref{4.47}). 

Similarly to the Miura transformation relating the mKdV equation to the KdV (see Figure \ref{fig1}), solutions of the KdV equation can easily be obtained from the solutions of (\ref{sidv+}) by using (\ref{m3}). Conversely, but very differently of the KdV-mKdV case, transformation (\ref{m3}) opens doors to construct solutions of (\ref{sidv+}) from the solutions of (\ref{m2}) by solving two linear equations! The aforementioned transformation, joint with the homogeneity of (\ref{sidv+}), verily brings a sort of linearity to the last equation, which makes the problem of finding solutions of (\ref{sidv+}) quite easier from theoretical point of view: several solutions of (\ref{m2}) are known and solving a linear equation is a procedure, overall, simpler than looking for solutions of nonlinear equations. These observations are expressed by system (\ref{4.47}): fixed a solution of the KdV equation, it is a linear system to $u$ and, by Theorem \ref{teo4.3}, if this solution is a non-vanishing one, then we have a solution to (\ref{sidv+}). It worth emphasising that, in the KdV-mKdV case, in order to obtain a solution of the mKdV from the KdV equation it is imperative to solve a nonlinear equation: more precisely, the Riccati equation showed in Figure \ref{fig1}. 

The applicability of Theorem \ref{teo4.3} is shown in examples \ref{ex4.1}, \ref{ex4.2} and subsection \ref{kink} as well. In the first example, from constant solutions of (\ref{m2}) we obtained exponential or periodic solutions of (\ref{sidv+}). The second example exhibits an interesting solution: the similarity solution $x/6t$ of the KdV equation leads to a solution of (\ref{sidv+}) in terms of the Airy functions as stated in (\ref{4.51}). Finally, subsection \ref{kink} explores a new solution of (\ref{sidv+}) obtained from the 1-soliton solution (\ref{6.1}) of (\ref{m2}): the result is a classical, but physically and mathematically relevant solution, given by the kink wave (\ref{6.1.4}). The relevance of the aforesaid function comes from the fact that it is a solitary wave.

While for the KdV equation (\ref{m2}) solution (\ref{6.1}) provides a wave travelling as fast as its amplitude is big, the kink wave (\ref{6.1.4}) has its amplitude unaltered but, on the other hand, it tends as fast to its asymptotic values as its phase velocity is big. This is a natural consequence due to the fact that at each point, its slope is proportional to the $\sech^2$ solution of the KdV equation and it increases with the phase velocity. These behaviors are showed in Figure \ref{fig8}.

\begin{center}
\begin{figure}[h!]
    \centering
    \begin{subfigure}[b]{0.45\textwidth}
       \includegraphics[width=\textwidth]{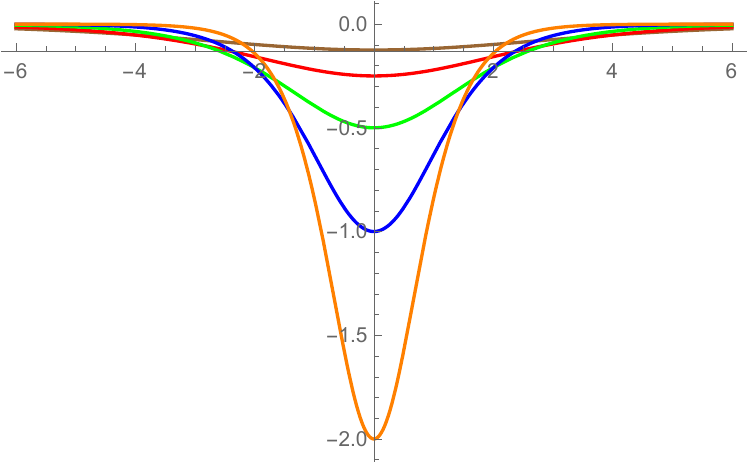}
        \label{fig:8a}
        \caption{}
    \end{subfigure}
    \begin{subfigure}[b]{0.45\textwidth}\label{fig:8b}
\includegraphics[width=\textwidth]{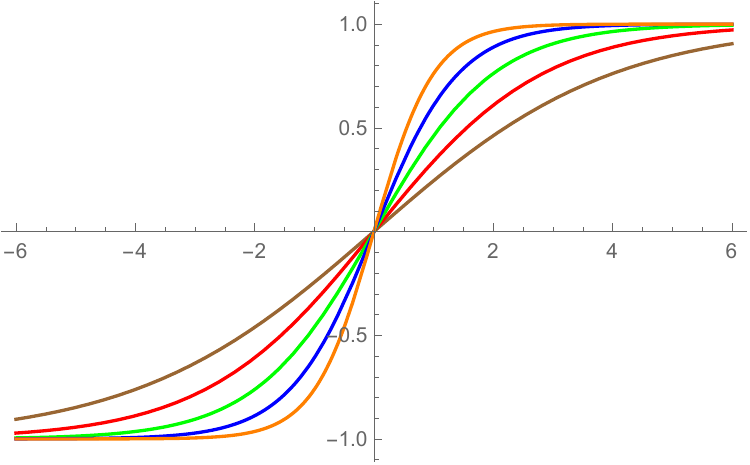}
        \caption{}    
                           \end{subfigure} 
      \caption{Behavior of solutions (\ref{6.1}) and (\ref{6.1.4}) at $t=0$ and different values of the phase velocity: colors brown, red, free, blue and orange show $w(x,0)$ $(a)$ and $u(x,0)$ $(b)$ for $c=0.25, 0.5, 1, 2$ and $4$, respectively. }\label{fig8}
    \end{figure}
\end{center}

Another featured solitary wave solution is the peakon, whose existence is ascertained by Theorem \ref{teo5.1}, although in this case this solution only appears when $\epsi=0$, which reduces (\ref{1.1}) to a second order evolution equation.

\section{Conclusion}\label{conclusion}

The goal of this paper is the investigation of properties related to equation (\ref{1.1}). To this end, we derived low order conservation laws in section \ref{cl}, some of them, new. 

The results of section \ref{integrability} supports the integrability of the cases $\epsi a=1$ and $\epsi=\pm2/3$ of equation (\ref{1.1}). Moreover, equation (\ref{miura}) enables one to determine a new type of Miura transformation, connecting (\ref{1.1}) with the KdV equation. In particular, we applied these ideas to find solutions of (\ref{sidv+}) as illustrated by examples \ref{ex4.1} and \ref{ex4.2} and the kink solution (\ref{6.1.4}). As far as we know, solution (\ref{4.51}) given in terms of the Airy functions, and the kink wave (\ref{6.1.4}), are new.

We prove that equation (\ref{1.1}) also admits sinusoidal and exponential peakons, depending on the sign of the quotient $c/a$. Moreover, this result corrects a previous one \cite{ijcnmac}, in which peakon solutions were believed to exist for any value of $\epsi$.

Finally, our main results are: Theorem \ref{main1}, which was needed to establish the conservation laws given in section \ref{cl}; theorems \ref{teo4.1} and \ref{teo4.2} which show that (\ref{1.1}) admits two integrable members; the Miura transformation (\ref{m3}), which, although its simplicity, not only answered a point raised in \cite{sen}, but also gives the condition to obtain solutions of (\ref{sidv+}) by means of a Schrödinger operator with potential given by the solutions of the KdV equation and, more interestingly and intriguing, brings a {\it certain} linearity to (\ref{sidv+}), as stated by Theorem \ref{teo4.3} and Corollary \ref{cor4.2}; Theorem \ref{teo5.1}, which shows that (\ref{1.1}) admits peakon solutions. Although this is achieved at a very particular case, as far as we know, it is a first time that peakon functions are reported as solutions of evolution equations.

\section*{Acknowledgements}

 The authors would like to thank FAPESP, grant nº 2014/05024-8, for financial support. P. L. da Silva is grateful to CAPES and FAPESP (grant nº 12/22725-4) for her scholarships. I. L. Freire is also grateful to CNPq for financial support, grant nº 308941/2013-6. J. C. S. Sampaio is grateful to CAPES and FAPESP (scholarship nº 11/23538-0) for financial support. The authors show their gratitude to Dr. M. Marrocos, Dr. E. A. Pimentel, Dr. J. F. S. Pimentel, Dr. F. Toppan and Dr. Z. Kuznetsova for their support and stimulating discussions. 

\end{document}